\documentclass[11pt,a4paper]{article}

\usepackage[truedimen,margin=30mm]{geometry} 

\usepackage{mathrsfs}
\usepackage{amssymb}
\usepackage{amsmath}
\usepackage{ascmac}
\usepackage{amsthm}

\usepackage[pdftex]{graphicx}       
\usepackage[pdftex]{color}     
\usepackage{natbib}
\usepackage{setspace}
\usepackage{color}
\usepackage{url}
\usepackage{booktabs}

\usepackage{enumerate}

\usepackage{bm}
\usepackage{comment}
\usepackage{pifont}%
\usepackage{titlesec}
\titleformat*{\section}{\large\bfseries}
\titleformat*{\subsection}{\it}

\newtheorem{thm}{Theorem}
\newtheorem{lem}{Lemma}

\def\ta{{\tau}}

\def\ep{{\varepsilon}}

\def\si{{\sigma}}

\def\pd{\partial}

\def\bta{{\text{\boldmath $\eta$}}}
\def\bla{{\text{\boldmath $\lambda$}}}

\def\bxi{{\text{\boldmath $\xi$}}}

\def\btau{{\text{\boldmath $\tau$}}}
\def\bmu{{\text{\boldmath $\mu$}}}

\def\bta{{\text{\boldmath $\eta$}}}

\def\bla{{\text{\boldmath $\lambda$}}}

\def\bxi{{\text{\boldmath $\xi$}}}

\def\btau{{\text{\boldmath $\tau$}}}
\def\bmu{{\text{\boldmath $\mu$}}}

\def\Si{{\Sigma}}

\def\bSi{{\text{\boldmath $\Si$}}}

\def\bSit{{\widetilde \bSi}}
\def\a{{\text{\boldmath $a$}}}

\def\e{{\text{\boldmath $e$}}}

\def\u{{\text{\boldmath $u$}}}
\def\v{{\text{\boldmath $v$}}}
\def\w{{\text{\boldmath $w$}}}

\def\z{{\text{\boldmath $z$}}}
\def\A{{\text{\boldmath $A$}}}
\def\B{{\text{\boldmath $B$}}}

\def\D{{\text{\boldmath $D$}}}
\def\E{{\text{\boldmath $E$}}}

\def\H{{\text{\boldmath $H$}}}
\def\I{{\text{\boldmath $I$}}}

\def\V{{\text{\boldmath $V$}}}

\def\X{{\text{\boldmath $X$}}}

\def\wbt{{\tilde \w}}

\def\diag{{\rm diag\,}}

\def\non{{\nonumber}}
\title{{\bf Locally Adaptive Bayesian Isotonic Regression using Half Shrinkage Priors}}

\date{}
\author{}

\begin{document}

\maketitle
\doublespacing

\vspace{-1.5cm}
\begin{center}
Ryo Okano$^1$, Yasuyuki Hamura$^2$, Kaoru Irie$^3$ and Shonosuke Sugasawa$^4$
\end{center}

\noindent
$^1$Graduate School of Economics, The University of Tokyo\\
$^2$Graduate School of Economics, Kyoto University\\
$^3$Faculty of Economics, The University of Tokyo\\
$^4$Faculty of Economics, Keio University

\maketitle

\vspace{1cm}
\begin{center}
{\bf \large Abstract}
\end{center}
Isotonic regression or monotone function estimation is a problem of estimating function values under monotonicity constraints, which appears naturally in many scientific fields.
This paper proposes a new Bayesian method with global-local shrinkage priors for estimating monotone function values. 
Specifically, we introduce half shrinkage priors for positive valued random variables and assign them for the first-order differences of function values.
We also develop fast and simple Gibbs sampling algorithms for full posterior analysis.
By incorporating advanced shrinkage priors, the proposed method is adaptive to local abrupt changes or jumps in target functions. 
We show this adaptive property theoretically by proving that the posterior mean estimators are robust to large differences and that asymptotic risk for unchanged points can be improved.  
Finally, we demonstrate the proposed methods through simulations and applications to a real data set. 

\bigskip\noindent
{\bf Key words}: Horseshoe prior; Isotonic regression; Monotone curve estimation; Tail robustness.

\newpage
\section{Introduction}

The problem of estimating function under monotonicity constraints is generally called isotonic regression or monotone function estimation \citep{barlow1972statistical, robertson1988order, groeneboom2014nonparametric}.
This problem naturally arises in a variety of scientific applications including the dose-response modeling in epidemiological studies \citep{morton2000additive} and estimation of the demand curves in economics \citep{quah2000monotonicity}.

There is a large literature on point estimation for isotonic regression, including the minimization of sums-of-squares errors subject to ordering constraints \citep{brunk1972statistical, robertson1988order}, the use of splines \citep{ramsay1998estimating, wang2008isotonic} and kernels \citep{mammen1991estimating, hall2001nonparametric} with appropriate constraints. 
Relaxing the hard constraint on monotonicity is known as the nearly-isotonic regression and has also been studied \citep{tibshirani2011nearly,matsuda2021generalized}. 
To extend those approaches to allow for uncertainty quantification by posterior analysis, we need to model a prior on the target function. This can be done in multiple ways, including the linear combination of monotone functions \citep{neelon2004bayesian,cai2007bayesian,mckay2011variable,shively2009bayesian}, Gaussian processes combined with monotone projection \citep{lin2014bayesian} and Gaussian processes for derivatives of functions \citep{lenk2017bayesian,kobayashi2021flexible}. Although these methods can realize the posterior distribution of the exactly or approximately monotone function, they have two limitations. First, the models are not flexible enough to express a function with ``jumps,'' or sudden increases/decreases at several points. Second,  posterior computation can be costly and inefficient due to the restriction of monotonicity.

In this paper, we address both problems by considering a prior for the first-order differences of the function of interest at consecutive points. In doing so, we utilize the global-local shrinkage by the half-horseshoe (HH) prior, a heavily-tailed distribution defined on the positive real line, to ensure the exact monotonicity. This distribution is the truncated version of the horseshoe prior \citep{carvalho2010horseshoe} and exhibits the same heavy tail, which enables the adaptive response of the posterior distribution to jumps. In addition, the HH prior can be written as the scale mixture of half-normals (the normal distribution truncated on the positive real line), allowing an efficient posterior computation by a simple Gibbs sampler. The horseshoe and its variants have already been applied to the first-order or higher-order differences in the context of time series analysis \citep{faulkner2018locally,bitto2019achieving,kowal2019dynamic} and shown the successful posterior inference for complex patterns of functional forms. We add the monotone constraint to this series of research.

We clarified the utility of the HH prior for the isotonic regression through the two theoretical results about the Bayes estimator of the target function. First, we verified the posterior tail-robustness under the proposed model. In other words, the more extreme jump we observe, the less shrinkage we apply a posteriori, making the Bayes estimator adaptive to the observed jump. Although this property is well-known for the horseshoe-type prior \citep{carvalho2010horseshoe,hamura2022global}, the evaluation of posterior robustness for isotonic regression is complicated due to the correlation between the function values at two consecutive points. Consequently, the robustness for the isotonic regression by the HH prior is ``weaker'' than those without monotonicity constraint, in terms of the speed of diminishing shrinkage effects. 
Second, we proved the Kullbuck-Leibler super-efficiency of the proposed method under sparsity, showing efficiency in estimating functions that are constant on some regions.

The rest of the paper is organized as follows. In Section~\ref{methods}, we provide the model formulation and introduce shrinkage priors for positive valued differences of coefficients. The Gibbs sampler for posterior computation is provided in Section~\ref{methods}.
In Section~\ref{theory}, we show theoretical results on the posterior robustness to jumps and the efficiency in handling sparsity. 
Section~\ref{simulation} is devoted to numerical experiments for the investigation of finite sample performance of the proposed methods together with some existing methods.
The application to the detection of the drop of Nile-river water flow is discussed in Section~\ref{data_analysis}. The proofs, algorithms, and additional experimental results are provided in the Appendix.

\section{Methods}
\label{methods}
\subsection{Model formulation}
We assume the $n$ observations $y_1, ..., y_n$ are conditionally independent and modeled by
$$
y_i=f(x_i)+\epsilon_i, \ \ \epsilon_i\sim N(0, \sigma^2), \quad \text{for $i=1, ..., n$}, 
$$
where $f: \mathbb{R} \to \mathbb{R}$ is monotone. 
We discuss the estimation of $f$ from the viewpoint of the Gaussian location model by writing $\theta_i=f(x_i)$. 
Then, the monotonicity constraint translates to $\theta_1 \le \theta_2 \le \cdots \le \theta_n$, and the model can be rewritten as 
\begin{equation}
y|\theta , \sigma ^2 \sim N(\theta , \sigma ^2 I_n) ,  
\label{model_1}
\end{equation}
where $y = (y_1, ..., y_n)^\top$ and $\theta = (\theta_1, ..., \theta_n)^\top$.

We guarantee the monotonicity of function values by using priors for the first-order differences of function values supported in $(0, \infty)$. 
Denote the differences by
$\eta_j = \theta_j - \theta_{j-1}$ for $j=2, ..., n$ and the initial value by $\eta_1 = \theta_1$. 
Then we have $\theta = D \eta$ where $\eta = (\eta_1, ..., \eta_n)^\top$ and $D$ is a $p\times p$ lower triangular matrix defined below: 
\begin{equation}
D = \begin{bmatrix}
1 & 0 & \cdots & 0 & 0 \\
1 & 1 & \cdots & 0 & 0 \\
& & \ddots & & \vdots \\
1 & 1 & \cdots & 1 & 0 \\
1 & 1 & \cdots & 1 & 1 
\end{bmatrix}.
\label{D_mat}
\end{equation}
In the next subsection, we specify the prior for $\eta$ to define the prior for $\theta$ implicitly via equality $\theta = D \eta$. 
The likelihood of $\eta$ is obtained as 
\begin{equation*}
y | \eta, \sigma^2\sim N(D\eta ,\sigma ^2 I_n).
\end{equation*}

\subsection{Priors for differences of coefficients} \label{priors}
To model positive parameters $\eta_2, ..., \eta_n$, we utilize global-local shrinkage priors whose supports are the half real line. 
We assume that $\eta _2,\dots ,\eta _n$ are mutually independent and follow the truncated normal distributions on $(0,\infty)$ as 
\begin{equation}
\eta_j|\tau_j^2, \lambda^2, \sigma^2 \sim N_{+}(0, \sigma^2\lambda^2\tau_j^2) \quad \text{and} \quad \tau_j \sim \pi(\tau_j), \quad \text{for $j=2, ..., n$},
\label{prior_form}
\end{equation}
where both $\lambda$ and $(\tau_2, ..., \tau_n)$ are all positive. The prior for $\eta$, or equivalently for $\theta$, is scaled by observational variance $\sigma^2$, hence we define our uncertainty on function values relative to the sampling variations. This density function is also convenient for posterior computation for its conditional conjugacy for $\sigma^2$.
 Note also that this prior does not depend on the arguments of a function $f$, or $(x_1,\dots ,x_n)$, for simplicity. Hence we implicitly assume that the function values are observed on evenly-spaced grids. For the extension to the case of irregular grids, see Section~\ref{grids}.

From the viewpoint of shrinkage estimation, global parameter $\lambda$ shrinks all $\eta_j$'s toward zero uniformly, while local parameter $\tau_j$ provides a custom shrinkage effect for each individual $j$. 
The flexibility of shrinkage effects depends on the choice of priors for the global and local shrinkage parameters, $\lambda$ and $(\tau_2, ..., \tau_n)$. 
The use of these parameters in (\ref{prior_form}) is known as the global-local shrinkage and imported from the literature on shrinkage priors \citep{carvalho2010horseshoe}. In the context of trend filtering, the same global-local shrinkage technique has been practiced (e.g., \citealt{faulkner2018locally}) for the estimation of $\eta_j$, where the half-normal distribution in (\ref{prior_form}) is replaced by the double-sided normal distribution. 

Here we discuss the choice of priors for $\tau_j$'s; the prior for $\lambda$ is provided later in Section~\ref{global}. Following the discussions of \cite{faulkner2018locally}, we consider three types of priors: half-horseshoe, half-Laplace and half-normal.

\begin{itemize}
	\item \textit{Half-horseshoe prior}.
	The horseshoe prior \citep{carvalho2009handling,carvalho2010horseshoe} has the property of a very strong shrinkage effect, as represented by the density spikes at near zero. In addition, its density has a Cauchy-like heavy tail, indicating that the shrinkage is not applied to outlying signals. This combination results in an excellent performance as a shrinkage prior. We assume a truncated version of horseshoe prior for $\eta_j$. 
	This is realized by assuming that $\tau_j$ follows the half-Cauchy distribution in (\ref{prior_form}) or, equivalently, $\tau_j^2|\nu_j \sim Ga(1/2, \nu_j)$ and $\nu_j \sim Ga(1/2, 1)$, where $Ga(1/2, \nu_j)$ is the Gamma distribution with shape $1/2$ and rate $\nu_j$ (with mean $1/(2\nu_j)$). 
	We name this model for $\eta_j$ the half-horseshoe prior.
	The density spike at $\eta _j=0$ can also be seen in this half-horseshoe prior, reflecting our belief that $\theta _j \approx \theta _{j+1}$ for most $j$'s. Likewise, the heavy tail is expected to explain a possible large deviation of $\theta_{j+1}$ from $\theta_j$.

	\item \textit{Half-Laplace prior}. 
	The half-Laplace prior, which is merely the exponential prior, is obtained by assuming that $\tau_j^2 \sim Ga(1, \nu)$. As the truncated version of the Bayesian Lasso \citep{park2008bayesian}, this model will be compared with the half-horseshoe prior.

	\item \textit{Half-normal prior}.
	We also consider the half-normal prior, by setting $\tau _j=1$ in (\ref{prior_form}). 
The normal prior does have a shrinkage effect as the Bayesian counterpart of the Ridge estimation, but its shrinkage effect and tail-robustness are much weaker than those of the horseshoe and Laplace priors. 
\end{itemize}

In the left panel of Figure~\ref{density}, we visually illustrate the marginal distributions of $\eta_j$ under the half-horseshoe, half-Laplace and half-normal priors. It is shown in this figure that the density under the half-horseshoe prior has a spike at the origin and heavier tail than those under the half-Laplace and half-normal priors. In addition, 
in the right panel of Figure~\ref{density}, we generate four sets of coefficients $(\theta_1, ..., \theta_{10})$ from the half-horseshoe prior with $\lambda^2 = 1$ and $\sigma^2 = 1$, and illustrate the shapes of the functions. As seen in those realizations of $f$, the HH prior represents a strong prior belief that the graph of $f$ is mostly flat and has occasional jumps. If the functional form of the true $f$ is unmatched to this description, then the posterior inference with the HH prior is likely to be unsuccessful. In the simulation study, we consider the linear, constantly increasing $f$ (Scenario (III) in Section~\ref{data_analysis}), where we confirm that the HL and HN priors perform better. Thus, the HH prior is not a multi-purpose model, but should be chosen with great care in applications.

In what follows, we focus on the study of this finite-dimensional, correlated shrinkage prior for $\theta$ for practical reasons. 
Viewing this prior as the marginal distribution of a stochastic process for function $f$ is of potential interest, but we do not pursue such development here. 
In the case of irregular grids, finding a corresponding stochastic process for $f$ is difficult (e.g., see Section 2.4 of \citealt{faulkner2018locally}).

\begin{figure}[htbp]
	\begin{minipage}[b]{0.45\linewidth}
		\centering
		\includegraphics[keepaspectratio, width=75mm]{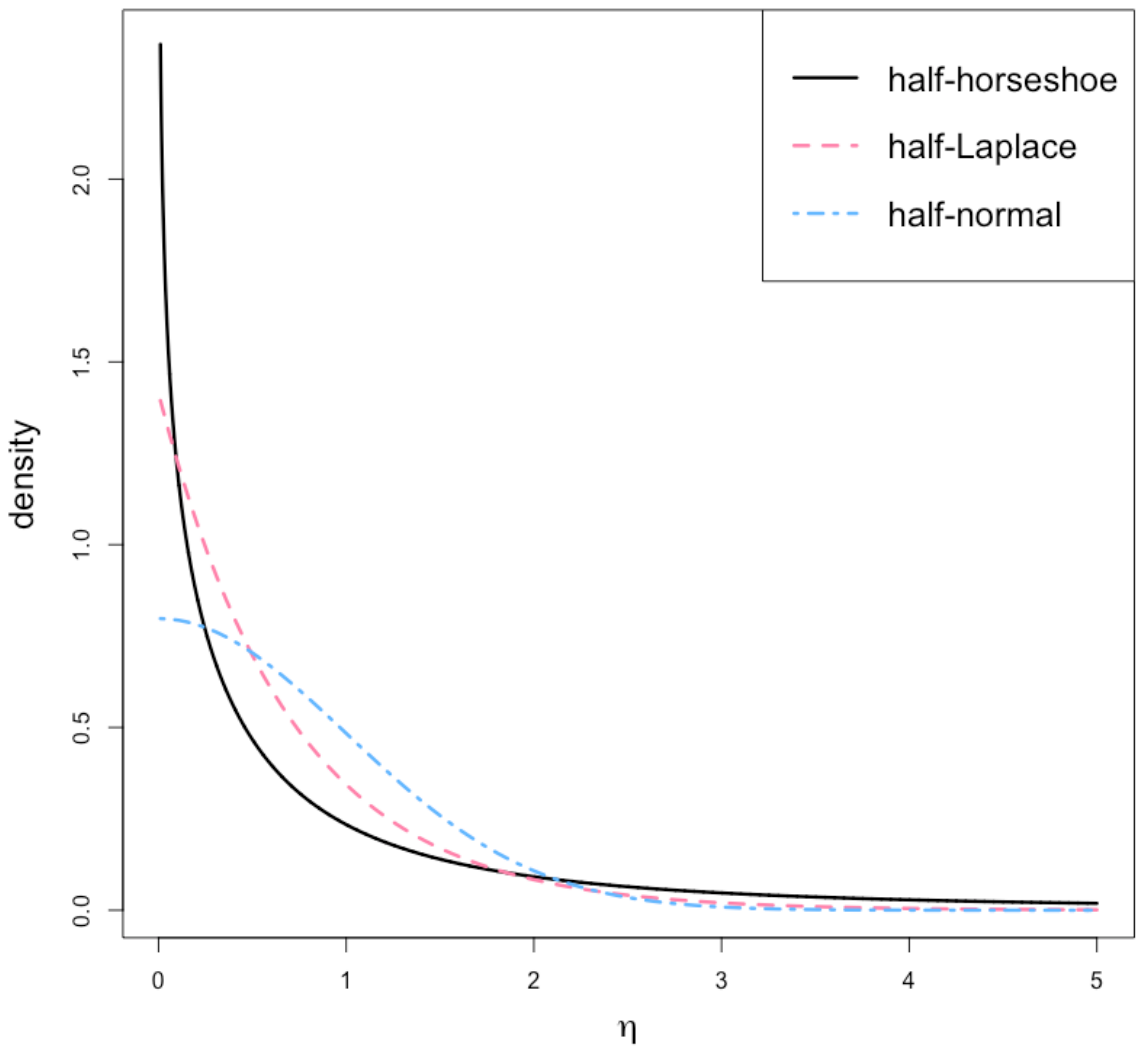}
	\end{minipage}
	\begin{minipage}[b]{0.45\linewidth}
		\centering
		\includegraphics[keepaspectratio, width=75mm]{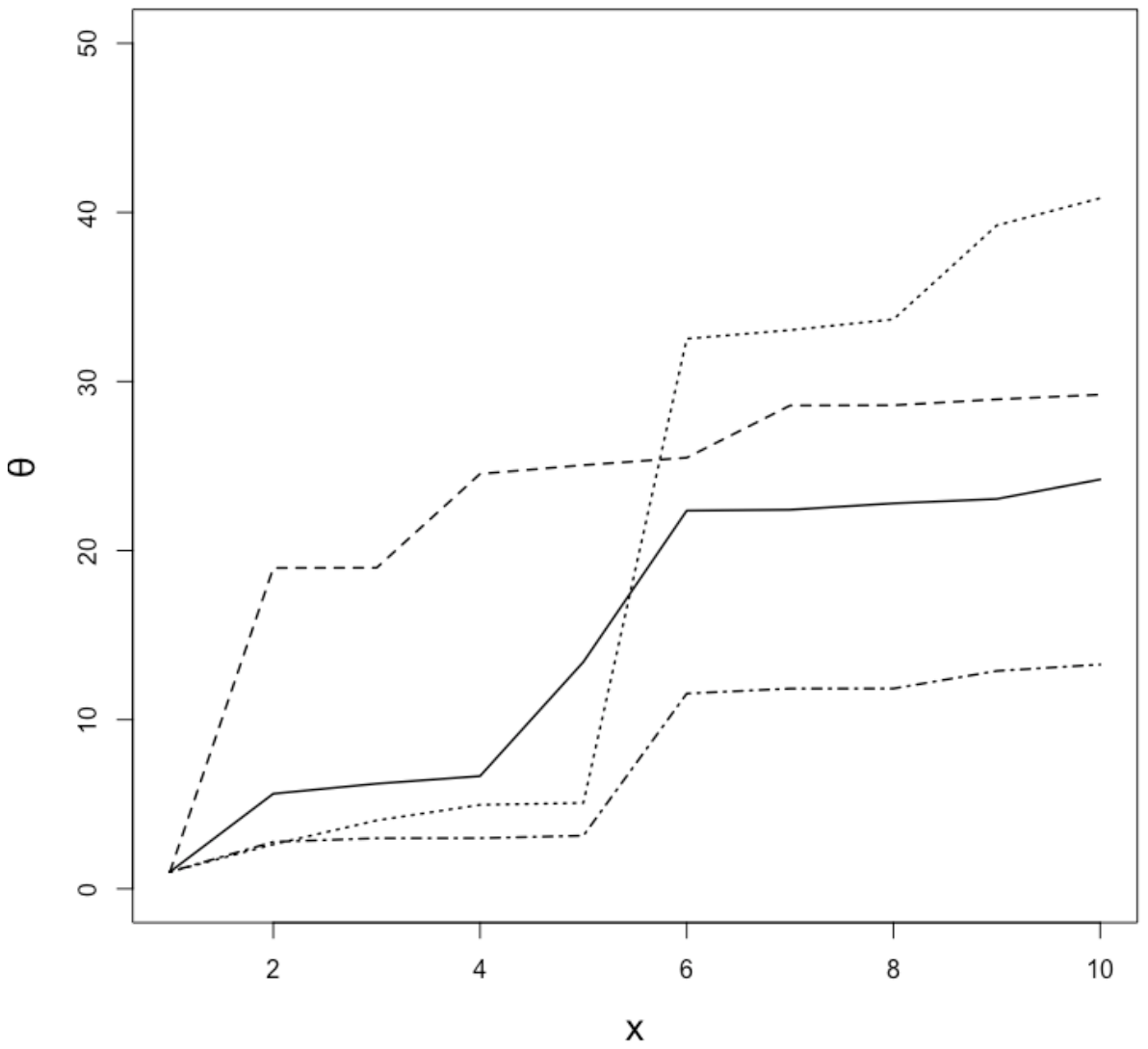}
	\end{minipage}
	\caption{Shapes of the half-horseshoe, half-Laplace and half-normal densities (left). Examples of shapes of monotone functions generated from the half-horseshoe prior (right).}
	\label{density}
\end{figure}

\subsection{Global parameters and posterior computation} \label{global}
For the observational variance $\sigma^2$, we assign the non-informative prior, $\pi(\sigma ^2)=\sigma^{-2}$. For the global shrinkage parameter, we set $\lambda^2|\xi \sim Ga(1/2, \xi)$ and $\xi \sim Ga(1/2, 1)$. 
We also set the scale-beta prior for $\eta_1$, namely, 
 $\eta _1 | \tau _1^2, \sigma ^2 \sim N( 0, \sigma ^2 \tau _1^2)$,  $\tau _1^2 | \nu _1 \sim Ga(1,\nu_1)$, and $\nu _1 \sim Ga(1/2,1).$ We use this set of priors for all the models, including the half-horseshoe, half-Laplace and half-normal models for $\eta _{2:n}$.

Although the Bayes estimator of $\theta_i$ is not analytically available, there is an efficient yet simple Markov chain Monte Carlo algorithm for posterior computation. 
The full conditional posteriors of all parameters become well-known distributions, so that we can efficiently carry out Gibbs sampling by generating posterior sampling from those distributions. We provide detailed step-by-step Gibbs samplings in the Appendix.

\subsection{Extension to irregular grids} \label{grids}
Our model formulation is so far restricted to the case where data is observed at equally spaced locations. Here, following \cite{faulkner2018locally}, we generalize the model to allow for data observed at irregularly spaced locations. 
Let $x_1 < x_2 < \cdots < x_n$ be the locations of observations, and we denote by $w_j = x_j - x_{j-1}$ the distance between adjacent locations. 
For the case of irregular grids, we assume the following prior for $\eta_j$:
\begin{equation}
\eta_j|\tau_j^2, \lambda^2, \sigma^2 \sim N_{+}(0, \sigma^2\lambda^2\tau_j^2 w_j) \quad \text{and} \quad \tau_j \sim \pi(\tau_j), \quad \text{for $j=2, ..., n$}.
\label{prior_irr}
\end{equation}
The difference from the regular grid case (\ref{prior_form}) can be seen in the dependence of the variance of $\eta_j$ on $w_j$. 
For this dependence, for example, when $w_j$ is small, it is expected to have a stronger shrinkage effect on $\eta_j = \theta_j - \theta_{j-1}$. 
The half-horseshoe, half-Laplace and half-normal priors for $\eta_j$ can be applied
in the same way with the case of the regular grid. The algorithm of the Gibbs sampler in the irregular grid case is trivially obtained and provided in the Appendix.

\section{Theoretical properties}
\label{theory}

\subsection{Posterior robustness for $\eta_j$}
The original horseshoe estimator for sparse signals is robust in the sense that the difference between the estimate and the observation vanishes as the observation becomes extreme \citep{carvalho2010horseshoe}. 
This property is known as tail-robustness.
In this subsection, we show that the posterior mean estimator in our setting is also robust to large positive signals, but in a different sense. 
This theoretical result explains the reason that the proposed method can make the posterior of the function adaptive to jumps.

Consider the model (\ref{model_1}), where $\sigma^2$ is assumed to be fixed. 
We treat the initial value $z_1 = y_1$ and the differences $z_j = y_j - y_{j-1}$ as observations and we rewrite the model (\ref{model_1}) as $z \sim N(\eta, \Sigma)$, where 
$z = D^{-1}y = (y_1, y_2 - y_1, ..., y_n - y_{n-1})^\top, \eta = D^{-1}\theta$ and $\Sigma = \sigma^2 D^{-1}(D^{-1})^\top$, and $D$ is the lower triangular matrix defined in (\ref{D_mat}). 
We consider the three distributions provided in Section~\ref{priors} as the priors for the location differences. 
That is, conditional on $\tau_j$'s, the prior distribution for $\eta_j$ for $j=2, ..., n$ is the scale mixture of half-normals (and the prior for $\eta_1$ is the scale mixture of normals):
\begin{gather*}
\eta_1 | \tau_1 \sim N(0, \tau_1^2), \\ 
\eta_j | \tau_j \sim N_+(0, \tau_j^2), \quad j=2, ..., n. 
\end{gather*}
where $\tau_j$ is fixed (half-normal), $\tau_j^2\sim Ga(1,\nu)$ for some $\nu>0$ (half-Laplace) or $\tau_j$ follows the half-Cauchy distribution (half-horseshoe). 
For simplicity, we omit the global shrinkage parameter.
We arbitrarily choose an index $i_\ast \in \{2, ..., n \}$, where a jump of the function value is observed, and fix the values $z_j \in \mathbb{R}$ for $j \neq i_\ast$. 
That is, $z_{i_\ast}$ is the only outlier in our setting. 
The next theorem shows the behavior of the posterior mean $E[\eta_{i_\ast} | z]$ in the presence of a large signal.

\begin{thm}
	The posterior means, $E[\eta_{i_\ast} | z]$, of the half-Laplace and half-horseshoe models described above are weakly tail-robust, in the sense that 
	\begin{equation*}
	\frac{ |E[\eta_{i_\ast} | z] - z_{i_\ast}| }{ z_{i_\ast} } \to 0 \quad \mathrm{as} \ z_{i_\ast} \to \infty.
	\end{equation*}
	\label{tail_robust}
\end{thm}

The proof for the half-Laplace model is surprisingly concise due to the simplicity of the marginal density of $\eta_j$. In contrast, the half-horseshoe density involves integral representation and complicates the proof. For details, see the Supplementary Materials. 

The robustness property shown above is weaker than the usual tail-robustness (i.e., $|E[\eta_{i_\ast} | z] - z_{i_\ast}| \to 0$ as $z_{i_\ast} \to \infty$). 
As stated in the introduction, the analysis of the posterior means under our setting is more difficult than a standard setting where the location parameters $\eta_j$'s are conditionally independent a posteriori \citep{carvalho2010horseshoe}. The difference from the standard setting, and the source of weakened robustness, is the correlation of observation $z$ (or, equivalently, the correlation of the original location parameter $\theta$). The robustness property of horseshoe estimators under such a correlation has not been studied. 

It is noteworthy that not only the half-horseshoe model but also the half-Laplace model can achieve the weak tail-robustness. The difference between the two models in point estimation is evaluated numerically in Section~\ref{simulation}, where outlier $z_{i_\ast}$ takes a finite value. The half-normal prior, whose density tail is thinner than that of the half-Laplace distribution, does not have the weak tail-robustness. To see this, consider a simple model for a single observation:  $z_j|\eta _j \sim N(\eta_j, \sigma^2)$ and $\eta_j \sim N(0, 1)$, ignoring both monotonicity constraint and correlation for simplicity. Then we obtain the posterior mean estimator $E[\eta_j | z_j] = (1 + \sigma^2)^{-1}z_j$, which does not have the weak tail-robustness since $| E[\eta_j | z_j] - z_j | / z_j \to \sigma^2 / (1+\sigma^2)$ as $z_j\to\infty$. 
This example, together with Theorem~\ref{tail_robust}, shows the necessity of the prior for $\eta_j$ with a heavier tail than the Gaussian one to achieve the weak tail-robustness.


%
\subsection{Efficiency in handling sparsity}
The half-horseshoe prior can also handle sparsity in $\eta_j$ efficiently for its density spike at the origin. 
Following \cite{carvalho2010horseshoe}, we formalize this efficiency in terms of the rate of convergence of the Kullback-Leibler divergence between the true sampling density and its posterior mean estimator. 
By showing this efficiency, we claim that under the sparse setting, the (half) horseshoe prior can reconstruct the true sampling density faster than any other prior with a bounded density at the origin. 
Unlike the weak-tail robustness, the proof completely parallels that of the standard horseshoe prior \citep[Theorem~4]{carvalho2010horseshoe}.

For any parameter $\eta \ge 0$, let denote the probability density function of $x$ under $\eta$ as $f_\eta(x) = N(x | \eta, 1), x \in \mathbb{R}$. 
Let $\eta_0$ denote the true value of the parameter, and let $\Delta_{KL}(\eta_0 || \eta) = \Delta_{KL}\{f_{\eta_0} || f_\eta \}$
denote the Kullback-Leibler divergence of $f_\eta$ from $f_{\eta_0}$.
Because $f_\eta$ is assumed to be the normal density, we can express the Kullback-Leibler divergence explicitly as 
\[
\Delta_{KL}(\eta_0 || \eta)
=
\frac{(\eta - \eta_0)^2}{2}.
\]
Suppose we have $n$ observations $y_1, ..., y_n$.
Let $p$ be the marginal prior density on $\eta$, and for each $k = 1,\dots,n$, let
$\pi_k(d\eta | y_1, ..., y_k)$ be the posterior density based on $k$ observations $y_1, ..., y_k$, respectively.
Then the posterior mean estimator of the density based on the $k$ observations is defined as
\[
\hat{f}_k(x)
=
\int_0^\infty f_\eta(x) \pi_k(d \eta | y_1, ..., y_k).
\]
As the accuracy of estimation, we use the following Ces\'{a}ro-average risk:
\[
R_n
=
\frac{1}{n} \sum_{k=1}^n \Delta_{KL}\{f_{\eta_0} || \hat{f}_k \}.
\]
The following theorem shows that the half-horseshoe prior leads to a super-efficient rate of convergence when $\eta_0 = 0$.
\begin{thm}
	Suppose the true sampling model is $N(\eta_0, 1)$, and we assume the half-horseshoe prior for $\eta_0$.
	Then, the rates of convergence of Ces\'{a}ro-average risk $R_n$ under $\eta_0 = 0$ and $\eta_0 > 0$ are given as follows
	\[
	R_n = \begin{cases}
	O[n^{-1}\{\log n - \log \log n \}], \quad \eta_0 = 0, \\
	O\{n^{-1}\log n \}, \quad \eta_0 > 0.
	\end{cases}
	\]
	\label{kl_risk}
\end{thm}

As stated in \cite{carvalho2010horseshoe}, the optimal convergence rate of Ces\'{a}ro-average risk is $O\{n^{-1}\log n \}$ when using a priors with bounded density. Therefore,
Theorem \ref{kl_risk} implies that under the sparse setting $\eta_0 = 0$, the density estimator based on the half-horseshoe prior converges to the true density faster than those based on other bounded priors on $[0, \infty)$.
Theorem \ref{kl_risk} also implies that for $\eta_0 > 0$, the density estimator based on the half-horseshoe prior should converge at least as fast as those based on other bounded priors.

\section{Simulation study}
\label{simulation}
We here investigate the finite sample performance of the proposed method together with other methods. We generated $n=100$ observations from $y_i \sim N(\theta_i, 0.25) $ where $\theta_i$ is the value of a monotone function $f(x)$ at $x=i$. In generating the true monotone function $f$, we adopted the following five scenarios: 
\begin{align*}
&\text{(I)}\,\ f(x) = 2, \\
&\text {(II)}\,\  f(x) = 
\begin{cases}
0 & (0 \le x \le 25)\\
2.5 & (25 < x \le 80)  \\
3 & (80 < x \le 100)
\end{cases}, \\
&\text{(III)}\,\ f(x) = 0.04x, \\
&\text{(IV)}\,\ f(x) = 
\begin{cases}
0.02x & (0 \le x \le 20)\\
0.02x + 1 & (20 < x \le 50)  \\
0.02x + 1.5 & (50 < x \le 80)  \\
0.02x + 1.75 & (80 < x \le 100)
\end{cases},  \\
&\text{(V)}\,\ f(x) = 
\frac{1}{4.4}\exp\{ 0.05x - 2\} + 1.
\end{align*}
These true functions are constant, piecewise constant, linear, piecewise linear and exponential shape, respectively. In scenarios (I) and (III), the true function has less frequent increments in comparison to the other scenarios. Note that in scenarios (II) and (IV), there are a few jumps in the function values. 

We estimated $\theta_i$'s using the following six methods:
\begin{itemize}
	\item HH: The proposed method with the half-horseshoe prior for $\eta_j$.
	\item HL: The method with the half-Laplace prior for $\eta_j$.
	\item HN: The method with the half-normal prior for $\eta_j$.
	\item TF: The Bayesian trend-filtering with the horseshoe prior proposed by \cite{faulkner2018locally}. 
	\item GP: The Gaussian process regression method (e.g., \cite{williams2006gaussian}). 
	\item GPP: The Gaussian process projection method proposed by \cite{lin2014bayesian}.
\end{itemize}

To implement the posterior analysis in the GP and GPP methods, we employed the Python package "GPflow". 
On the covariance kernel of Gaussian process, we assigned the same prior used in the simulation in \cite{lin2014bayesian}.
All the methods require computations by  Markov chain Monte Carlo methods. For each dataset, we generated 2,500 posterior samples after discarding 500 samples as a burn-in period. 
We computed the posterior means of $\theta_i$'s as their point estimates in all methods. Note that the estimated function is guaranteed to be monotonic in the HH, HL, HN and GPP methods, while it is not necessarily monotonic in the TF and GP methods. 
The performance of these point estimators was evaluated by the root mean squared error (RMSE) defined as $\sqrt{n^{-1} \sum_{i=1}^n (\hat{\theta}_i - \theta_i)^2}$. We also computed $95\%$ credible intervals of $\theta_i$'s and evaluated their performance using the coverage probability (CP) and average length (AL). We repeated this process for $1,000$ times.

In Table \ref{table_regular}, we presented the averaged values of the RMSEs, CPs and ALs in all the methods. We can see that under scenarios (I), (III) and (V), the GPP method performs well in terms of RMSE, CP and AL, while the proposed HH method is comparable in RMSE. For scenarios (II) and (IV), where the true function has a few jumps, the HH method performs the best in RMSE, as well as achieves high coverage rates by the interval estimates of reasonable lengths. Seeing this result, we confirm that the
proposed half-horseshoe prior can flexibly represent functions with abrupt changes, as expected.

In addition, we reported an example of posterior fits under three methods HH, TF and GPP in Figure~\ref{example_fits_1}. 
We can see that the HH method is most adaptive to the abrupt increases in scenarios (II) and (IV).
The same graphical illustration for the other three methods (the HL, HN and GP methods) is reported in the Appendix.
Here, we presented the posterior plots in Scenarios (II) only in Figure~\ref{example_fits_3}. 
We can see that the HH method is more adaptive to abrupt increases in function values than the HL method.
This result shows that, although both HL and HH models achieve the weak tail-robustness in theory (Theorem~\ref{tail_robust}), the horseshoe model can respond to jumps more quickly for its heavier density tail.

\begin{table}[!htbp]
	\begin{center}
		\begin{tabular}{cclcccccc} \toprule
			Scenario &       & HH    & HL     & HN   & TF    & GP   & GPP  \\ \midrule
			& RMSE  & 0.034  & 0.044   &0.057  & 0.094  & 0.030 & 0.021    \\
			(I)     & CP    &91.9   &90.3    &83.4  &99.7   &98.1  &99.0    \\
			& AL    &0.121   &0.127    &0.151  &0.469   &0.138  &0.118 \\ \midrule
			& RMSE  &0.087   &0.312    &0.517  &0.115   &0.222  &0.200  \\
			(II)    & CP    &91.9   &90.0    &83.9  &98.6   &91.3  &90.4  \\
			& AL    &0.220   &0.345    &0.411  &0.501   &0.535  &0.315  \\ \midrule
			& RMSE  &0.079   &0.055    &0.057  &0.090   &0.05  &0.05  \\
			(III)   & CP    &98.8   &99.4    &99.5  &99.3   &94.0  &94.4  \\
			& AL    &0.357   &0.284    &0.257  &0.467   &0.193  &0.190  \\ \midrule
			& RMSE  &0.081   &0.139    &0.166  &0.127   &0.156  &0.141  \\
			(IV)   & CP    &97.1   &77.1    &66.7  &96.6   &81.0  &82.2  \\
			& AL    &0.322   &0.291    &0.275  &0.493   &0.357  &0.321  \\ \midrule
			& RMSE  &0.110   &0.159    &0.231  &0.094   &0.075  &0.072  \\
			(V)  & CP    &89.8   &65.8    &43.0  &99.5   &94.0  &94.2  \\
			& AL    &0.352   &0.308    &0.311  &0.460   &0.279  &0.256 \\  \bottomrule
		\end{tabular}
	\end{center}
	\caption{Averaged values of root mean squared errors (RMSEs) of point estimators, and coverage probability (CP) and average length(AL) of $95\%$ credible intervals under five scenarios with $n=100$.}
	\label{table_regular}
\end{table}

\begin{figure}[htbp]
	\centering
	\includegraphics[keepaspectratio, width=5.00in]{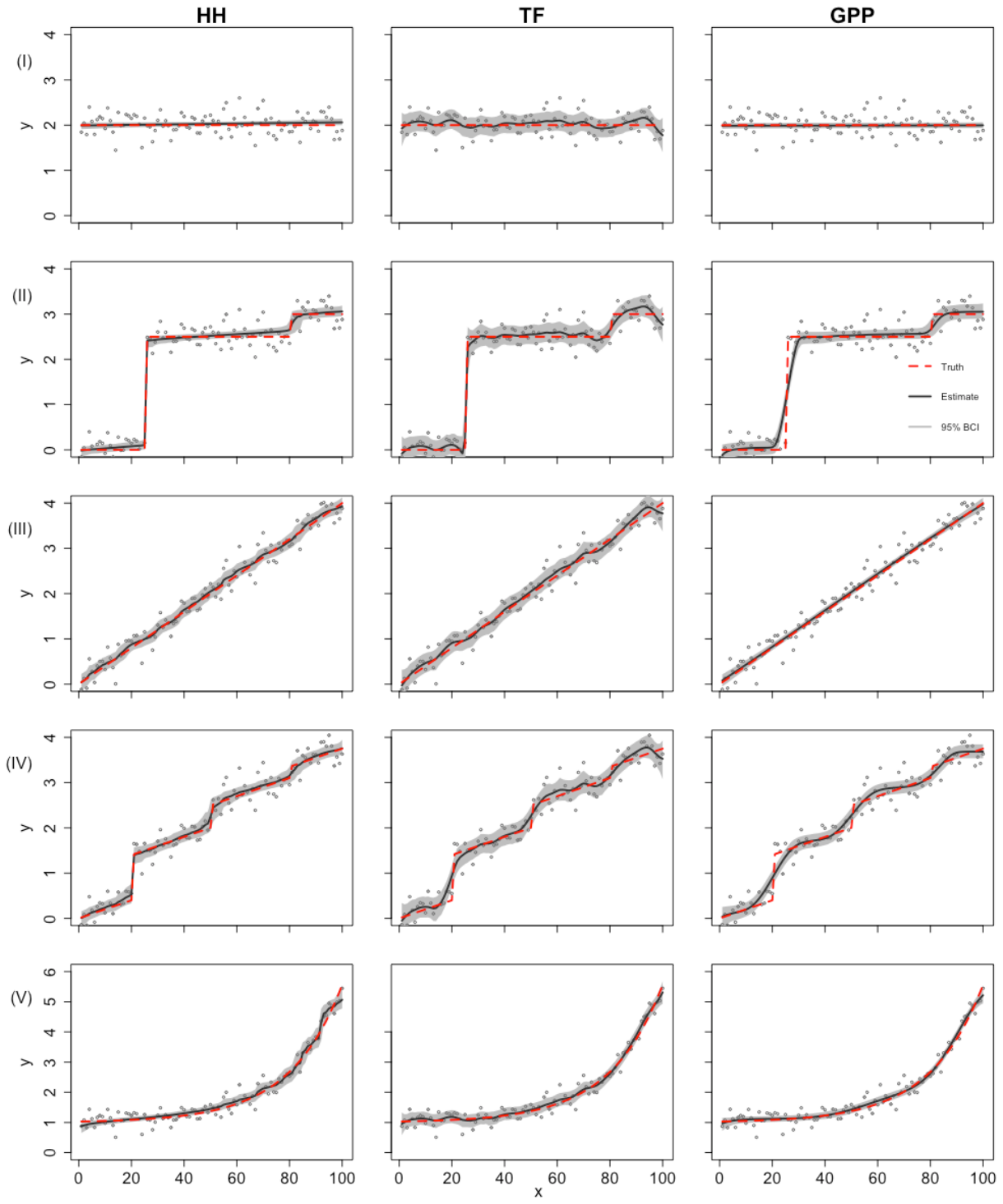}
	\caption{Example fits by the HH, TF and GPP methods under the five scenarios. Plots show true functions (dashed red lines), posterior means(solid black lines), and associated $95 \%$ Bayesian credible interval (gray bands) for each $\theta_i$. Values between observed locations are interpolated for plotting.}
	\label{example_fits_1}
\end{figure}

\begin{figure}[htbp]
	\centering
	\includegraphics[keepaspectratio, width=5.00in]{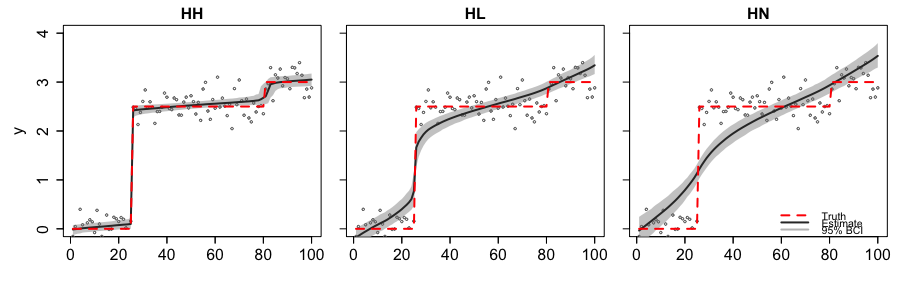}
	\caption{Example fits by the HH, HL and HN methods under the scenarios (II). Plots show true functions (dashed red lines), posterior means(solid black lines), and associated $95 \%$ Bayesian credible interval (gray bands) for each $\theta_i$. Values between observed locations are interpolated for plotting.}
	\label{example_fits_3}
\end{figure}

Next, we present a simulation result when data is observed at irregularly spaced locations. We randomly picked 25 elements from $\{1, 2, ..., 100\}$ and sorted them as $x_1 < x_2 < \cdots < x_{25}$. We generated 25 observations from $y_k \sim N(f(x_k), 0.25), k=1, ..., 25$ under the five scenarios for the true monotone function. We estimated the 100 function values $f(i)$, $i=1, ..., 100$ from the 25 observations $y_k$'s, including those without the corresponding observations, by using the HH, GP and GPP methods. To interpolate $f(x)$ for $x\in \{ x_1,\dots ,x_{25}\}$, we first estimated $f(x_k)$'s by the method described in Section 2.4, then computed the function values between the observed locations by the linear interpolation.
In Table \ref{table:irregular}, we reported the averaged values of RMSEs, CPs and ALs in this setting. 
We can observe a result similar to the case of regular grids. For scenarios (II) and (IV), the HH method performs quite well in terms of RMSE and has short credible intervals with good coverage probabilities.

\begin{table}[!htbp]
	\centering
	\begin{tabular}{cclccc} \toprule
		Scenario &       & HH    & GP   & GPP  \\ \midrule
		& RMSE  &0.084   &0.056  &0.042  \\
		(I)   & CP    &88.0   &100.0  &99.3  \\
		& AL    &0.278   &0.300  &0.278  \\ \midrule
		& RMSE  &0.262   &0.385  &0.364  \\
		(II)  & CP    &81.6   &91.5  &89.3  \\
		& AL    &0.420   &1.147  &0.762  \\ \midrule
		& RMSE  &0.139   &0.097  &0.095  \\
		(III) & CP    &96.0   &98.4  &98.5  \\
		& AL    &0.566   &0.445  &0.446  \\ \midrule
		& RMSE  &0.178   &0.243  &0.224  \\
		(IV)  & CP    &92.2   &78.4  &78.5  \\
		& AL    &0.556   &0.605  &0.566  \\ \midrule
		& RMSE  &0.194   &0.210  &0.204  \\
		(V)   &CP     &92.1   &96.0  &96.4  \\
		& AL    &0.592   &0.700  &0.615  \\ \bottomrule
	\end{tabular}
	\caption{The averaged values of the RMSEs, CPs and ALs in the case of irregular grids.}
	\label{table:irregular}
\end{table}

\section{Data analysis}
\label{data_analysis}
In this section, we estimate the trend of the
yearly volume of the Nile River at Aswan by applying the proposed method. 
Our data consists of measurements of the annual flow of the Nile River from $1871$ to $1970$, obtained from \cite{durbin2012time} and publicly available in $\mathsf{R}$-package \textit{datasets} \citep{datasets}. 
This data has a decreasing trend as a whole and has an apparent change point near $1898$ when a dam was built. 
In this study, we are interested in whether or not the estimated trends can detect the change point. 
To this end, we apply the proposed isotonic regression method with the HH prior to this data, as well as the HL, HN and GPP methods. 
In each Gibbs sampler, we generate $5,000$ posterior samples after discarding $1,000$ posterior samples as burn-in.

Figure \ref{nile_plots} shows the posterior means and 95\% credible intervals of the function values. 
It is observed in this figure that the proposed method with the half-horseshoe prior successfully detects the abrupt change in $1898$. In contrast, this change is less emphasized, or not at all detected by the other methods. This is another example where the sufficient flexibility of the proposed method is utilized to allow large jumps.

\begin{figure}[htbp]
	\centering
	\includegraphics[keepaspectratio, width=5.65in]{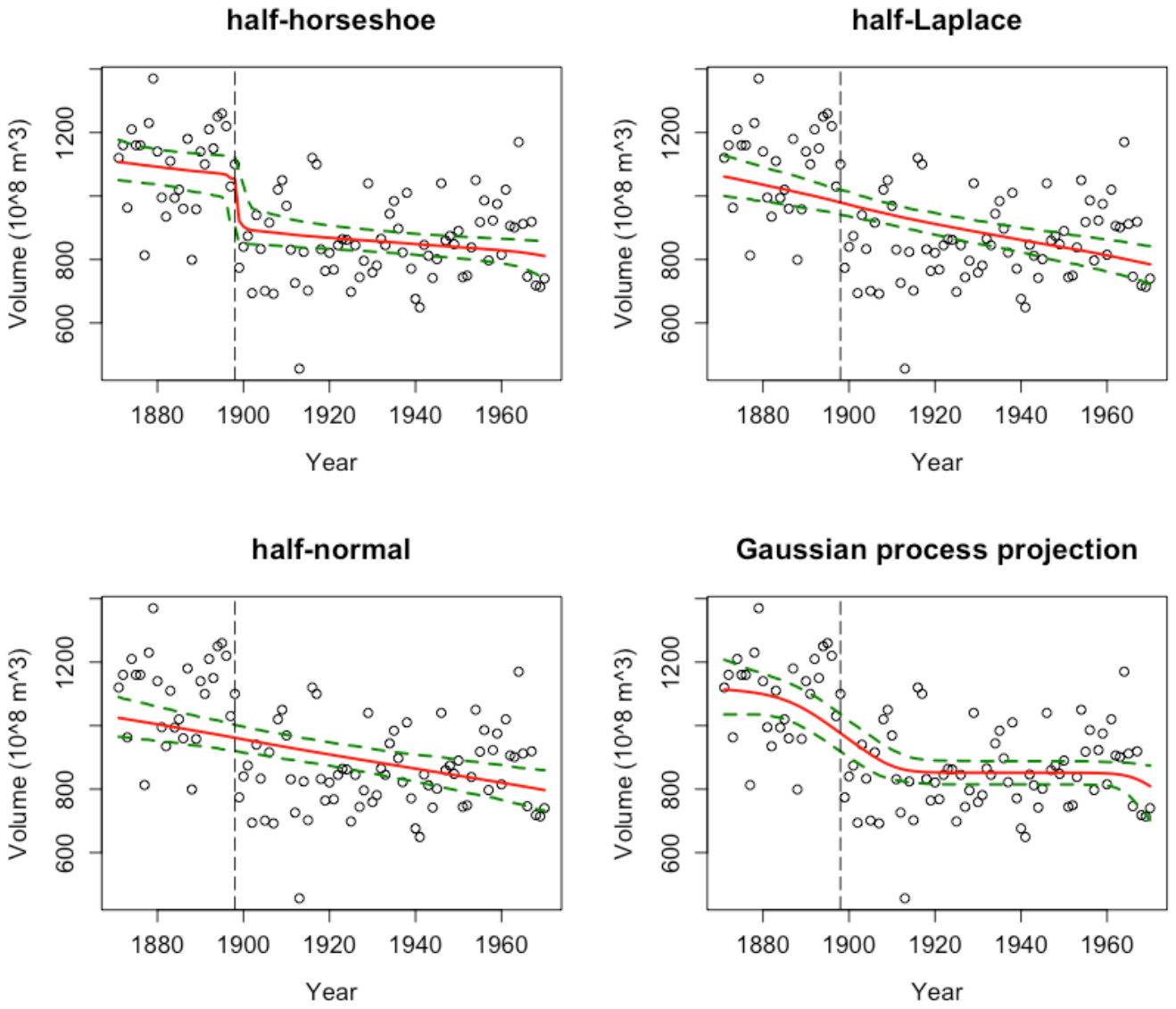}
	\caption{Estimates for the Nile river data by different methods. Solid lines are posterior means and dashed lines are $95\%$ credible intervals. Dashed vertical lines indicate the changepoint (1898).}
	\label{nile_plots}
\end{figure}

\section{Concluding remarks}
In this research, we proposed the half-horseshoe prior for the robust posterior inference of monotone functions. It should be noted that the half-horseshoe prior is a continuous distribution, where we place no probability mass on $\{ \theta _j = \theta _{j+1} \}$, or $\{\eta_j = 0\}$. The extension to the point-mass mixture, known as the spike-and-slab prior, is of great interest. It is well known, however, that the ability of handling sparsity under the spike-and-slab prior comes at the increased computational cost for the full posterior inference. This problem is expected to be inherited from the problem of estimating monotone functions. The trade-off of the model flexibility and computational cost is left for future research.

The application of the half-horseshoe prior to the higher-order differences is trivial. 
Although the first-order difference is extensively studied in this research to focus on the monotonicity, the second-order difference of function values is also of great interest to realize convex constraints on the target function. 
Investigating the methodological and theoretical properties of such methods for convex function estimation would be an important topic of future research.

In addition, the proposed method can be extended easily to situations having multiple covariates. 
For example, when $p$ covariates $(x_{i1},\ldots,x_{ip})$ are obtained, we may consider an additive model, $y_i=\sum_{k=1}^p f_k(x_{ik})+\epsilon_i$. 
If any functions of $f_1,\dots ,f_p$ are known to be monotone, then we can model the monotonicity by assigning the proposed prior independently for those functions. 
The posterior computation for such an additive model is straightforwardly implemented by Bayesian backfitting \citep{hastie2000bayesian}, combined with the conditional conjugacy of the proposed half-horseshoe priors.

\clearpage 
\appendix
\vspace{1cm}
\begin{center}
{\bf \large Appendix}
\end{center}

\section*{A1 \,\ Step-by-step Gibbs sampling procedures}
We here provide Gibbs sampling producers under the proposed half-horseshoe prior in both the regular and irregular grid cases.
We denote the GIG distribution with parameter $a > 0, b > 0, p \in \mathbb{R}$ by $GIG(a, b, p)$. That is, the density of $GIG(a, b, p)$ is given by
\[
f(x) = \frac{(a/b)^{p/2}}{2K_p(\sqrt{ab})} x^{p-1} e^{-(ax + b/x)/2}, 
\quad x > 0. 
\]
We also 
denote the inverted gamma distribution with shape parameter $\alpha > 0$ and scale parameter $\beta > 0$ by $IG(\alpha, \beta)$. That is, the density of $IG(\alpha, \beta)$ is 
\[
f(x) = \frac{\beta^\alpha}{\Gamma(\alpha)}\left(\frac{1}{x}\right)^{\alpha + 1}
e^{-\beta / x}, \quad x > 0.
\]
For $j = 1, ..., n$, let $d_j$ be the $j$-th columns vector of the lower triangular matrix $D$ and define a vector $e_j \in \mathbb{R}^n$ by
\[
e_j = y - \sum_{k \neq j}\eta_k d_k.
\]
$\|\cdot \|$ denotes the Euclidean norm.
The summary of the posterior sampling under the half-horseshoe prior in the general case (irregular grids) is as follows:
\begin{itemize}
	\item[-]  {\bf (Sampling from function values)} \ \  Sample $\eta_1$ from $N(m_1, s_1^2)$ where
	\[
	m_1 = \frac{e_1^\top d_1}{ \|d_1\|^2 + (1/\tau_1^2) }, \ \ \ \mathrm{and} \ \ \ s_1^2 = \frac{\sigma ^2}{\|d_1\| ^2 + (1/ \tau _1^2)},
	\]
	and sample $\eta_j (j=2, ..., n)$ from $N_+(m_j, s_j^2)$, where
	\[
	m_j = \frac{e_j^\top d_j}{ \|d_j\|^2 + \{1/\lambda ^2\tau_j^2(x_j - x_{j-1})\} }, \ \ \ \mathrm{and} \ \ \ s_j^2 = \frac{\sigma ^2}{\|d_j\| ^2 + \{1/ \lambda ^2\tau_j^2(x_j - x_{j-1})\}}.
	\]
	\item[-] {\bf (Sampling from local shrinkage parameters)} \ \  
	 First, sample $\nu_1$ from $Ga(3/2, 1+\tau_j^2)$, and sample $\nu_j (j=2, ..., n)$ from $Ga(1, 1+\tau_j^2)$.
	Then, sample  $\tau _1^2$ from $GIG(2\nu_1, \eta_1^2/\sigma^2,1/2)$, and sample $\tau_j^2 (j=2,\ldots,n)$ from $GIG(2\nu_j, \eta_j^2 / \{\sigma^2 \lambda ^2(x_j - x_{j-1})\}, 0)$.
	\item[-] {\bf (Sampling from scale parameter)} \ \ 
	 First, sample $\xi$ from $Ga(1, 1 + \lambda^2)$, and then sample $\lambda^2$ from 
	\[
	GIG\left(2\xi, \frac{1}{\sigma^2}\sum_{j=2}^n \frac{\eta_j^2}{\tau_j^2(x_j - x_{j-1})}, \frac{2-n}{2}\right).
	\]
	\item[-] {\bf (Sampling from error variance)} \ \  
	 Sample $\sigma^2$ from
	\[
	IG\left(n, \frac{1}{2}\left( \|y - D\eta\|^2 + \frac{1}{\lambda^2}\sum_{j=2}^n \frac{\eta_j^2}{\tau_j^2(x_j - x_{j-1})} + \frac{\eta_1^2}{\tau_1^2} \right)\right).
	\]
\end{itemize}
Note that the sampling algorithm under the regular grid can be obtained by setting $x_j=j$.

\section*{A2 \,\ Proof of Theorem 
1}

\subsection*{Half-Laplace priors}

We marginalize local scale parameters $\tau_j$ out and work on the following model: 
\begin{align}
&\z \sim {\rm{N}}_n ( \z | \bta , \bSi ) \text{,} \non \\
&\bta \sim e^{- | \eta _1 | / (2 c_1 )} \prod_{i = 2}^{n} e^{- \eta _i / (2 c_i )} \text{,} \non 
\end{align}
for $c_1 , \dots , c_n > 0$, where $\bta \in \mathbb{R} \times (0, \infty )^{n - 1}$. 
Fix $i_{*} = 2, \dots , n$. 
Then, 
as $z_{i_{*}} \to \infty $ by the dominated convergence theorem, we have 
\begin{align}
E[ \eta _{i_{*}} | \z ] - z_{i_{*}} &= \frac{ \displaystyle \int_{\mathbb{R} \times (0, \infty )^{n - 1}} ( \eta _{i_{*}} - z_{i_{*}} ) \exp \Big( - {1 \over 2} \Big[ ( \bta - \z )^{\top } \bSi ^{- 1} ( \bta - \z ) + {| \eta _1 | \over c_1} + \sum_{i = 2}^{n} {\eta _i \over c_i} \Big] \Big) d\bta }{ \displaystyle \int_{\mathbb{R} \times (0, \infty )^{n - 1}} \exp \Big( - {1 \over 2} \Big[ ( \bta - \z )^{\top } \bSi ^{- 1} ( \bta - \z ) + {| \eta _1 | \over c_1} + \sum_{i = 2}^{n} {\eta _i \over c_i} \Big] \Big) d\bta } \non \\
&= \frac{ \displaystyle \int_{[(- \z ) + \mathbb{R} \times (0, \infty )^{n - 1}]} \xi _{i_{*}} \exp \Big( - {1 \over 2} \Big[ \bxi ^{\top } \bSi ^{- 1} \bxi + {| \xi _1 + z_1 | \over c_1} + \sum_{i = 2}^{n} {\xi _i \over c_i} \Big] \Big) d\bxi }{ \displaystyle \int_{[(- \z ) + \mathbb{R} \times (0, \infty )^{n - 1}]} \exp \Big( - {1 \over 2} \Big[ \bxi ^{\top } \bSi ^{- 1} \bxi + {| \xi _1 + z_1 | \over c_1} + \sum_{i = 2}^{n} {\xi _i \over c_i} \Big] \Big) d\bxi } \non \\
&\to \frac{ \displaystyle \int_{\mathbb{R} \times [(- \z _{- i_{*}}) + \mathbb{R} \times (0, \infty )^{n - 2}]} \xi _{i_{*}} \exp \Big( - {1 \over 2} \Big[ \bxi ^{\top } \bSi ^{- 1} \bxi + {| \xi _1 + z_1 | \over c_1} + \sum_{i = 2}^{n} {\xi _i \over c_i} \Big] \Big) d( \xi _{i_{*}} , \bxi _{- i_{*}} ) }{ \displaystyle \int_{\mathbb{R} \times [(- \z _{- i_{*}}) + \mathbb{R} \times (0, \infty )^{n - 2}]} \exp \Big( - {1 \over 2} \Big[ \bxi ^{\top } \bSi ^{- 1} \bxi + {| \xi _1 + z_1 | \over c_1} + \sum_{i = 2}^{n} {\xi _i \over c_i} \Big] \Big) d( \xi _{i_{*}} , \bxi _{- i_{*}} ) } \non 
\end{align}
Since the right-hand side is finite, we obtain 
\begin{align}
\lim_{z_{i_{*}} \to \infty } {E[ \eta _{i_{*}} | \z ] - z_{i_{*}} \over z_{i_{*}}} = 0 \text{.} \non 
\end{align}

\subsection*{Half-horseshoe priors}

Recall that we consider the following model:
\begin{gather*}
\bm{z} \sim N(\bm{z} | \bm{\eta}, \Sigma), \\
\bm{\eta} \sim N(\eta_1 | 0, \tau_1^2) \prod_{i=2}^n N_+(\eta_i | 0, \tau_i^2), \\
\bm{\tau} \sim \prod_{i=1}^n \pi(\tau_i) = \prod_{i=1}^n \frac{1}{1 + \tau_i^2}, 
\end{gather*}
where $\bm{z} = ( z_i )_{i = 1}^{n}$, $\bm{\eta} = ( \eta _i )_{i = 1}^{n}$, and $\bm{\tau} = ( \ta _i )_{i = 1}^{n}$. 
Here we describe the notations used in the following proof. 
For an integer $k$, we denote the origin of $\mathbb{R} ^k$ and the $k \times k$ identity matrix as $\bm{0}^{(k)}$ and $\bm{I}_k$, respectively. For $a_1, ..., a_k \in \mathbb{R}$,  
$\text{diag}(a_1, ..., a_k)$ denotes the diagonal matrix whose $(i, i)$ element is $a_i$. $\bm{e}_i^{(k)}$ denotes the vector in $\mathbb{R}^k$ with a $1$ in the $i$ th coordinate and $0's$ elsewhere.
For a vector $\a \in \mathbb{R}^k$, $\a _{-j}$ denotes the vector with length $k-1$ obtained by removing the $j$ th coordinate of $a$. 
For a set $A$, we denote its cardinality as $|A|$. 
For a set $A \subset \mathbb{R}^k$ and $c \in \mathbb{R}^k$, we define a set $A + c = \{a + c: a \in A \}$. For a set $A \subset \mathbb{R}^k$ and $x \in \mathbb{R}^k$, 
$1(x \in A)$ is $1$ if $x \in A$, and $0$ otherwise.
$N_k( \bmu , \V )$ denotes the $k$-variate Gaussian distribution with mean $\bmu $ and covariance matrix $\V $, and for $B \subset \mathbb{R}^k$, $P(N_k( \bmu , \V ) \in B) $ means the probability $P( \X \in B)$ for $\X \sim N_k ( \bmu , \V )$. 
$a \equiv b$ denotes $b$ is defined by $a$.
$f(x) \sim g(x)$ denotes the function $f$ is asymptotically equal to the function $g$.
We will prove the claim in six steps. 

\noindent
\textbf{Step 1: Representation of posterior mean.}
First, we give a representation of the posterior mean that characterizes the prior's tail robustness in terms of a score function.
Let $m$ be the marginal density of $\bm{z}$.
Since
\begin{align*}
&m( \z ) = \int_{\mathbb{R} \times (0, \infty )^{n - 1}} {\rm{N}}_n ( \z | \bta , \bSi ) p( \bta ) d\bta \text{,} \\
&{\frac{\pd m( \z )}{\pd \z}}  = \int_{\mathbb{R} \times (0, \infty )^{n - 1}} \{ - {\bSi }^{- 1} ( \z - \bta ) \} {\rm{N}}_n ( \z | \bta , \bSi ) p( \bta ) d\bta \text{,} 
\end{align*}
we have the representation
\begin{align*}
E[\bm{\eta} | \bm{z}] - \z &= \bSi {\frac{1}{m( \z )}} {\frac{\pd m( \z )}{\pd \z}}.
\end{align*}
In the following, we focus on the asymptotic behavior of the score function, namely, $\{ 1 / m( \z ) \} \{ \pd / ( \pd \z ) \} m( \z )$.

\noindent
\textbf{Step 2: First reduction.}
Define the diagonal matrix  $\bm{T} = \diag ( \ta _1 , \dots , \ta _n )$.
We have
\begin{align*}
	&m( \z ) \notag \\
	&\propto \int_{(0, \infty )^n} \Big( \int_{\mathbb{R} \times (0, \infty )^{n - 1}} \Big\{ \prod_{i = 1}^{n} {\pi ( \ta _i ) \over \ta _i} \Big\} \exp \Big[ - {1 \over 2} \{ ( \bta - \z )^{\top } \bSi ^{- 1} ( \bta - \z ) + \bta ^{\top } \bm{T} ^{- 2} \bta \} \Big] d\bta \Big) d\bm{\tau} \\
	&= \int_{(0, \infty )^n} \Big( \Big\{ \prod_{i = 1}^{n} {\pi ( \tau _i ) \over \tau _i} \Big\} \exp \Big[ - {1 \over 2} \z ^{\top } \bSi ^{- 1} \{ \bSi - ( \bSi ^{- 1} + \bm{T} ^{- 2} )^{- 1} \} \bSi ^{- 1} \z \Big]   \\
	&\quad \times \int_{\mathbb{R} \times (0, \infty )^{n - 1}} \exp \Big[ - {1 \over 2} \notag \\
	&\quad \times \{ \bta - ( \bSi ^{- 1} + \bm{T} ^{- 2} )^{- 1} \bSi ^{- 1} \z \} ^{\top } ( \bSi ^{- 1} + \bm{T} ^{- 2} ) \{ \bta - ( \bSi ^{- 1} + \bm{T} ^{- 2} )^{- 1} \bSi ^{- 1} \z \} \Big] d\bta \Big) d\btau \\
	&\propto \int_{(0, \infty )^n} \Big( \Big\{ \prod_{i = 1}^{n} {\pi ( \tau _i ) \over \tau _i} \Big\} \exp \Big[ - {1 \over 2} \z ^{\top } \bSi ^{- 1} \{ \bSi - ( \bSi ^{- 1} + \bm{T} ^{- 2} )^{- 1} \} \bSi ^{- 1} \z \Big] {1 \over | \bSi ^{- 1} + \bm{T} ^{- 2} |^{1 / 2}}   \\
	&\quad \times \int_{\mathbb{R} \times (0, \infty )^{n - 1}} {\rm{N}}_n ( \bta | ( \bSi ^{- 1} + \bm{T} ^{- 2} )^{- 1} \bSi ^{- 1} \z , ( \bSi ^{- 1} + \bm{T} ^{- 2} )^{- 1} ) d\bta \Big) d\btau \\
	&= \int_{(0, \infty )^n} \Big( \Big\{ \prod_{i = 1}^{n} {\pi ( \tau _i ) \over \tau _i} \Big\} \exp \Big[ - {1 \over 2} \z ^{\top } \bSi ^{- 1} \{ \bSi - ( \bSi ^{- 1} + \bm{T} ^{- 2} )^{- 1} \} \bSi ^{- 1} \z \Big] {1 \over | \bSi ^{- 1} + \bm{T} ^{- 2} |^{1 / 2}}   \\
	&\quad \times \int_{(0, \infty )^{n - 1}} {\rm{N}}_{n - 1} ( \bta _{- 1} | \E _2 ( \bSi ^{- 1} + \bm{T} ^{- 2} )^{- 1} \bSi ^{- 1} \z , \E _2 ( \bSi ^{- 1} + \bm{T} ^{- 2} )^{- 1} {\E _2}^{\top } ) d{\bta _{- 1}} \Big) d\btau \\
	&\equiv \int_{(0, \infty )^n} \Big( \Big\{ \prod_{i = 1}^{n} {\pi ( \tau _i ) \over \tau _i} \Big\} \exp \Big[ - {1 \over 2} \z ^{\top } \bSi ^{- 1} \{ \bSi - ( \bSi ^{- 1} + \bm{T} ^{- 2} )^{- 1} \} \bSi ^{- 1} \z \Big] {1 \over | \bSi ^{- 1} + \bm{T} ^{- 2} |^{1 / 2}} \notag \\
	&\quad \times F( \btau ; \z ) \Big) d\btau ,  
	\end{align*}
where we set $\E _2 = ( \bm{0} ^{(n - 1)} , \I _{n - 1} ) \in \mathbb{R} ^{(n - 1) \times n}$ and
\[
F( \btau ; \z )
=
\int_{(0, \infty )^{n - 1}} {\rm{N}}_{n - 1} ( \bta _{- 1} | \E _2 ( \bSi ^{- 1} + \bm{T} ^{- 2} )^{- 1} \bSi ^{- 1} \z , \E _2 ( \bSi ^{- 1} + \bm{T} ^{- 2} )^{- 1} {\E _2}^{\top } ) d{\bta _{- 1}} .
\]
Notice that
\begin{align*}
&\bSi ^{- 1} \{ \bSi - ( \bSi ^{- 1} + \bm{T} ^{- 2} )^{- 1} \} \bSi ^{- 1} \notag \\
&= \bSi ^{- 1} [ \bSi - \{ \bSi - \bSi \bm{T} ^{- 1} ( \I _n + \bm{T} ^{- 1} \bSi \bm{T} ^{- 1} )^{- 1} \bm{T} ^{- 1} \bSi \} ] \bSi ^{- 1}   \\
&= \bm{T} ^{- 1} ( \I _n + \bm{T} ^{- 1} \bSi \bm{T} ^{- 1} )^{- 1} \bm{T} ^{- 1}   
\end{align*}
and
\begin{align*}
\Big\{ \prod_{i = 1}^{n} {\pi ( \tau _i ) \over \tau _i} \Big\} {1 \over | \bSi ^{- 1} + \bm{T} ^{- 2} |^{1 / 2}} = {\prod_{i = 1}^{n} \pi ( \tau _i ) \over | \I _n + \bm{T} \bSi ^{- 1} \bm{T} |^{1 / 2}}.  
\end{align*}
Combining these facts, we have
\begin{align*}
&m( \z ) \notag \\
&\propto \int_{(0, \infty )^n} {\prod_{i = 1}^{n} \pi ( \tau _i ) \over | \I _n + \bm{T} \bSi ^{- 1} \bm{T} |^{1 / 2}} \exp \Big\{ - {1 \over 2} \z ^{\top } \bm{T} ^{- 1} ( \I _n + \bm{T} ^{- 1} \bSi \bm{T} ^{- 1} )^{- 1} \bm{T} ^{- 1} \z \Big\} F( \btau ; \z ) d\btau . 
\end{align*}
Therefore, we have the following equality
\begin{align*}
	&{1 \over m( \z )} {\pd m( \z ) \over \pd \z } \notag \\
	&= \int_{(0, \infty )^n} \Big[ {\prod_{i = 1}^{n} \pi ( \tau _i ) \over | \I _n + \bm{T} \bSi ^{- 1} \bm{T} |^{1 / 2}} \exp \Big\{ - {1 \over 2} \z ^{\top } \bm{T}
	^{- 1} ( \I _n + \bm{T} ^{- 1} \bSi \bm{T} ^{- 1} )^{- 1} \bm{T} ^{- 1} \z \Big\} \\
	&\quad \times \Big\{ - \bm{T} ^{- 1} ( \I _n + \bm{T} ^{- 1} \bSi \bm{T} ^{- 1} )^{- 1} \bm{T} ^{- 1} \z F( \btau ; \z ) + {\pd F( \btau ; \z ) \over \pd \z } \Big\} \Big] d\btau \\
	&\quad / \int_{(0, \infty )^n} {\prod_{i = 1}^{n} \pi ( \tau _i ) \over | \I _n + \bm{T} \bSi ^{- 1} \bm{T} |^{1 / 2}} \exp \Big\{ - {1 \over 2} \z ^{\top } \bm{T} ^{- 1} ( \I _n + \bm{T} ^{- 1} \bSi \bm{T} ^{- 1} )^{- 1} \bm{T} ^{- 1} \z \Big\} F( \btau ; \z ) d\btau . 
\end{align*}

\noindent
\textbf{Step 3: Change of variables.}
Let $\D = \diag (| z_1 |, \dots , | z_n |)$.
We define a new variable $\bm{v} = (v_1, ..., v_n)^{\top}$ via $\btau = \D \v $, and put $\w = \D ^{- 1} \z $. 
We express the above display using new variables as
\begin{align}
	&{1 \over m( \z )} {\pd m( \z ) \over \pd \z }   \notag \\
	&= \int_{(0, \infty )^n} \Big[ {\prod_{i = 1}^{n} \pi (| z_i | v_i ) \over | \I _n + \D \V \bSi ^{- 1} \V \D |^{1 / 2}} \notag \\
	&\quad \times \exp \Big\{ - {1 \over 2} \w ^{\top } \V ^{- 1} ( \I _n + \D ^{- 1} \V ^{- 1} \bSi \V ^{- 1} \D ^{- 1} )^{- 1} \V^{- 1} \w \Big\} \notag \\
	&\quad \times \Big\{ - \D ^{- 1} \V ^{- 1} ( \I _n + \D ^{- 1} \V ^{- 1} \bSi \V ^{- 1} \D ^{- 1} )^{- 1} \V ^{- 1} \w F( \D \v ; \z ) + {\pd F \over \pd \z } ( \D \v ; \z ) \Big\} \Big] d\v \notag \\
	&\quad / \int_{(0, \infty )^n} \Big[ {\prod_{i = 1}^{n} \pi (| z_i | v_i ) \over | \I _n + \D \V \bSi ^{- 1} \V \D |^{1 / 2}} \notag \\
	&\quad \times \exp \Big\{ - {1 \over 2} \w ^{\top } \V ^{- 1} ( \I _n + \D ^{- 1} \V ^{- 1} \bSi \V ^{- 1} \D ^{- 1} )^{- 1} \V^{- 1} \w \Big\} \Big] F( \D \v ; \z ) d\v , 
	\label{basic_rep}
\end{align}
where $\V = \diag ( v_1 , \dots , v_n )$. 

\noindent
\textbf{Step 4: Calculation of limits.}
We choose an index $i_\ast \in \{2, ..., n\}$ and fix the values $z_j \in \mathbb{R}$ for $j \neq i_\ast$.
We calculate the limits of quantities that appear in  (\ref{basic_rep}) as $0 < z_{i_\ast} \to \infty$.

\noindent
\textbf{Step 4.1.}
First, we calculate the limits of ${\prod_{i = 1}^{n} \pi (| z_i | v_i ) / | \I _n + \D \V \bSi ^{- 1} \V \D |^{1 / 2}} $ and $\exp \{ - {1 \over 2} \w ^{\top } \V ^{- 1} ( \I _n + \D ^{- 1} \V ^{- 1} \bSi \V ^{- 1} \D ^{- 1} )^{- 1} \V^{- 1} \w \} $. Since
\begin{align*}
&\begin{vmatrix} \I _n & \D \V \\ \V \D & - \bSi \end{vmatrix} = (- 1)^n | \bSi | | \I _n - \D \V (- \bSi )^{- 1} \V \D| = | \I _n | | - \bSi - \V \D {\I _n}^{- 1} \D \V |,   
\end{align*}
we have the equality
$
| \I _n + \D \V \bSi ^{- 1} \V \D | = | \bSi + \V \D ^2 \V | / | \bSi |.   
$
Therefore, we have 
\begin{align*}
{\prod_{i = 1}^{n} \pi (| z_i | v_i ) \over | \I _n + \D \V \bSi ^{- 1} \V \D |^{1 / 2}} &\sim {\big\{ \prod_{i \neq i_{\ast }} \pi (| z_i | v_i ) \big\} \pi (| z_{i_{\ast }} |) / {v_{i_{\ast }}}^2 \over | \bSi + \V \D ^2 \V |^{1 / 2} / | \bSi |^{1 / 2}} \notag \\
&= {\big\{ \prod_{i \neq i_{\ast }} \pi (| z_i | v_i ) \big\} \pi (| z_{i_{\ast }} |) / {v_{i_{\ast }}}^2 \over | \D | | \D ^{- 1} \bSi \D ^{- 1} + \V ^2 |^{1 / 2} / | \bSi |^{1 / 2}}   \\
&\sim {\big\{ \prod_{i \neq i_{\ast }} \pi (| z_i | v_i ) \big\} \pi (| z_{i_{\ast }} |) / {v_{i_{\ast }}}^2 \over | \D | | \bSit + \V ^2 |^{1 / 2} / | \bSi |^{1 / 2}},    
\end{align*}
where we put $\bSit = \lim_{z_{i_\ast} \to \infty } \D ^{- 1} \bSi \D ^{- 1}$.
Furthermore, we have
\begin{align*}
&\exp \Big\{ - {1 \over 2} \w ^{\top } \V ^{- 1} ( \I _n + \D ^{- 1} \V ^{- 1} \bSi \V ^{- 1} \D ^{- 1} )^{- 1} \V^{- 1} \w \Big\} \notag \\
&= \exp \Big\{ - {1 \over 2} \w ^{\top } ( \V ^2 + \D ^{- 1} \bSi \D ^{- 1} )^{- 1} \w \Big\}   \\
&\sim \exp \Big\{ - {1 \over 2} \w ^{\top } ( \V ^2 + \bSit )^{- 1} \w \Big\} .   
\end{align*}

\noindent
\textbf{Step 4.2.}
Next, we calculate the limit of $F( \D \v ; \z )$.
Note that $F( \btau ; \z ) $ can be expressed as
\begin{align*}
&F( \btau ; \z ) = \int_{(0, \infty )^{n - 1}} {\rm{N}}_{n - 1} ( \bta _{- 1} | \E _2 ( \bSi ^{- 1} + \bm{T} ^{- 2} )^{- 1} \bSi ^{- 1} \z , \E _2 ( \bSi ^{- 1} + \bm{T} ^{- 2} )^{- 1} {\E _2}^{\top } ) d{\bta _{- 1}} \\
&= P( {\rm{N}}_{n - 1} ( \E _2 ( \bSi ^{- 1} + \bm{T} ^{- 2} )^{- 1} \bSi ^{- 1} \z , \E _2 ( \bSi ^{- 1} + \bm{T} ^{- 2} )^{- 1} {\E _2}^{\top } ) \in (0, \infty )^{n - 1} )   \\
&= P( {\rm{N}}_{n - 1} ( \bm{0} ^{(n - 1)} , \E _2 ( \bSi ^{- 1} + \bm{T} ^{- 2} )^{- 1} {\E _2}^{\top } ) \in (0, \infty )^{n - 1} - \E _2 ( \bSi ^{- 1} + \bm{T} ^{- 2} )^{- 1} \bSi ^{- 1} \z )   \\
&= P( {\rm{N}}_{n - 1} ( \bm{0} ^{(n - 1)} , \E _2 ( \bSi ^{- 1} + \bm{T} ^{- 2} )^{- 1} {\E _2}^{\top } ) \in \E _2 ( \bSi ^{- 1} + \bm{T} ^{- 2} )^{- 1} \bSi ^{- 1} \z - (0, \infty )^{n - 1} ).   
\end{align*}
In addition, since 
\begin{align*}
\E _2 ( \bSi ^{- 1} + \bm{T} ^{- 2} )^{- 1} \bSi ^{- 1} \z &= \E _2 ( \I _n + \bSi \bm{T} ^{- 2} )^{- 1} \z   \\
&= \E _2 \{ \I _n - ( \bm{T} ^{- 1} + \bm{T} \bSi ^{- 1} )^{- 1} \bm{T} ^{- 1} \} \z   \\
&= \E _2 \z - \E _2 ( \bm{T} ^{- 1} + \bm{T} \bSi ^{- 1} )^{- 1} \bm{T} ^{- 1} \z,   
\end{align*}
we have the equality
\begin{align*}
\E _2 ( \bSi ^{- 1} + \D ^{- 2} \V ^{- 2} )^{- 1} \bSi ^{- 1} \z &= \E _2 \z - \E _2 ( \D ^{- 1} \V ^{- 1} + \D \V \bSi ^{- 1} )^{- 1} \V ^{- 1} \w   \\
&= \E _2 \z - \E _2 ( \D ^{- 2} \V ^{- 1} + \V \bSi ^{- 1} )^{- 1} \V ^{- 1} \D ^{- 1} \w.   
\end{align*}
Therefore, we can represent $F( \D \v ; \z )$ as 
\begin{align}
&F( \D \v ; \z ) \notag \\
&= P( {\rm{N}}_{n - 1} ( \bm{0} ^{(n - 1)} , \E _2 ( \bSi ^{- 1} + \D ^{- 2} \V ^{- 2} )^{- 1} {\E _2}^{\top } ) \notag \\
&\quad \in \E _2 ( \bSi ^{- 1} + \D ^{- 2} \V ^{- 2} )^{- 1} \bSi ^{- 1} \z - (0, \infty )^{n - 1} )   \notag \\
&= P( {\rm{N}}_{n - 1} ( \bm{0} ^{(n - 1)} , \E _2 ( \bSi ^{- 1} + \D ^{- 2} \V ^{- 2} )^{- 1} {\E _2}^{\top } ) \notag \\
&\quad \in \E _2 \z - \E _2 ( \D ^{- 2} \V ^{- 1} + \V \bSi ^{- 1} )^{- 1} \V ^{- 1} \D ^{- 1} \w - (0, \infty )^{n - 1} )   \notag \\
&= {1 / (2 \pi )^{(n - 1) / 2} \over | \E _2 ( \bSi ^{- 1} + \D ^{- 2} \V ^{- 2} )^{- 1} {\E _2}^{\top } |^{1 / 2}}   \notag \\
&\quad \times \int_{\mathbb{R} ^{n - 1}} \Big( 1( \bta _{- 1} \in \E _2 \z - \E _2 ( \D ^{- 2} \V ^{- 1} + \V \bSi ^{- 1} )^{- 1} \V ^{- 1} \D ^{- 1} \w - (0, \infty )^{n - 1} )   \notag \\
&\quad \times \exp \Big[ - {1 \over 2} {\bta _{- 1}}^{\top } \{ \E _2 ( \bSi ^{- 1} + \D ^{- 2} \V ^{- 2} )^{- 1} {\E _2}^{\top } \} ^{- 1} \bta _{- 1} \Big] \Big) d{\bta _{- 1}} .   
\label{rep_F}
\end{align}
Since
\begin{align*}
\exp \Big[ - {1 \over 2} {\bta _{- 1}}^{\top } \{ \E _2 ( \bSi ^{- 1} + \D ^{- 2} \V ^{- 2} )^{- 1} {\E _2}^{\top } \} ^{- 1} \bta _{- 1} \Big] \le \exp \Big\{ - {1 \over 2} {\bta _{- 1}}^{\top } ( \E _2 \bSi {\E _2}^{\top } )^{- 1} \bta _{- 1} \Big\} , 
\end{align*}
we can apply the dominated convergence theorem in the right-hand side of (\ref{rep_F}). 
As a result, we obtain the following limit 
\begin{align*}
&F( \D \v ; \z )   \\
&\to {1 / (2 \pi )^{(n - 1) / 2} \over | \E _2 [ \bSi ^{- 1} + \{ \A ( \v _{- i_{\ast }} ) \} ^2 ]^{- 1} {\E _2}^{\top } |^{1 / 2}}   \\
&\quad \times \int_{\mathbb{R} ^{n - 1}} \Big[ 1( \bta _{- (1, i_{\ast } )} \in \z _{- (1, i_{\ast } )} - \E _{- (1, i_{\ast } )} ( \B ( \v _{- i_{\ast }} ) + \V \bSi ^{- 1} )^{- 1} \V ^{- 1} \wbt - (0, \infty )^{n - 2} )   \\
&\quad \times \exp \Big\{ - {1 \over 2} {\bta _{- 1}}^{\top } ( \E _2 [ \bSi ^{- 1} + \{ \A ( \v _{- i_{\ast }} ) \} ^2 ]^{- 1} {\E _2}^{\top } )^{- 1} \bta _{- 1} \Big\} \Big] d{\bta _{- 1}} \\
&= P( {\rm{N}}_{n - 2} ( \bm{0} ^{(n - 2)} , \E _{- (1, i_{\ast } )} [ \bSi ^{- 1} + \{ \A ( \v _{- i_{\ast }} ) \} ^2 ]^{- 1} {\E _{- (1, i_{\ast } )}}^{\top } ) \\
&\quad \in \z _{- (1, i_{\ast } )} - \E _{- (1, i_{\ast } )} \{ \B ( \v _{- i_{\ast }} ) + \V \bSi ^{- 1} \} ^{- 1} \V ^{- 1} \wbt - (0, \infty )^{n - 2} ),   
\end{align*}
where we put 
\begin{gather*}
\A ( \v _{- i_{\ast }} ) = \lim_{z_{i_\ast} \to \infty } \D ^{- 1} \V ^{- 1} \text{,} \quad \B ( \v _{- i_{\ast }} ) = \lim_{z_{i_\ast} \to \infty } \D ^{- 2} \V ^{- 1} \text{,} \quad \wbt = \lim_{z_{i_\ast} \to \infty } \D ^{- 1} \w 
\end{gather*}
and where $\bta _{- (1, i_{\ast } )} = ( \eta _i )_{i \in \{ 1, \dots , n \} \setminus \{ 1, i_{\ast } \} }$, $\z _{- (1, i_{\ast } )} = ( z_i )_{i \in \{ 1, \dots , n \} \setminus \{ 1, i_{\ast } \} }$, and $\E _{- (1, i_{\ast } )} = (( \e _{i}^{(n)} )_{i \in \{ 1, \dots , n \} \setminus \{ 1, i_{\ast } \} } )^{\top }$.

\noindent
\textbf{Step 4.3.}
Finally, we derive the limit of $\{ ( \pd F) / ( \pd \z ) \} ( \D \v ; \z )$ as $z_{i_\ast} \to \infty$. We put 
$\u = \E _2 ( \bSi ^{- 1} + \bm{T} ^{- 2} )^{- 1} \bSi ^{- 1} \z $ and rewrite $\{ ( \pd F) / ( \pd z_{i} ) \} ( \btau ; \z)$ as
\begin{align*}
	&{\pd F( \btau ; \z ) \over \pd z_{i}} \notag \\
	&= \sum_{k = 1}^{n - 1} {\pd u_k \over \pd z_{i}} {\pd \over \pd u_k} \int_{- \infty }^{u_k} \Big\{ \notag \\
	&\quad \int_{\u _{- k} - (0, \infty )^{n - 2}} {\rm{N}}_{n - 1} ( \bta _{- 1} | \bm{0} ^{(n - 1)} , \E _2 ( \bSi ^{- 1} + \bm{T} ^{- 2} )^{- 1} {\E _2}^{\top } ) d{\bta _{- (1, k + 1)}} \Big\} d{\eta_{k + 1}}   \\
	&= \sum_{k = 1}^{n - 1} \Big\{ ( \e _{k}^{(n - 1)} )^{\top } \E _2 ( \bSi ^{- 1} + \bm{T} ^{- 2} )^{- 1} \bSi ^{- 1} \e _{{i}}^{(n)} \\
	&\quad \times\int_{\u _{- k} - (0, \infty )^{n - 2}} {\rm{N}}_{n - 1} (( u_{k} , \bta _{- (1, k + 1)} ) | \bm{0} ^{(n - 1)} , \E _2 ( \bSi ^{- 1} + \bm{T} ^{- 2} )^{- 1} {\E _2}^{\top } ) d{\bta _{- (1, k + 1)}} \Big\} \\
	&= \sum_{k = 1}^{n - 1} \Big\{ ( \e _{k}^{(n - 1)} )^{\top } \E _2 ( \bSi ^{- 1} + \bm{T} ^{- 2} )^{- 1} \bSi ^{- 1} \e _{{i}}^{(n)} \notag \\
	&\quad \times {\rm{N}} ( u_k | 0, ( \e _{k}^{(n - 1)} )^{\top } \E _2 ( \bSi ^{- 1} + \bm{T} ^{- 2} )^{- 1} {\E _2}^{\top } \e _{k}^{(n - 1)} )   \\
	&\quad \times \int_{\u _{- k} - (0, \infty )^{n - 2}} {{\rm{N}}_{n - 1} (( u_{k} , \bta _{- (1, k + 1)} ) | \bm{0} ^{(n - 1)} , \E _2 ( \bSi ^{- 1} + \bm{T} ^{- 2} )^{- 1} {\E _2}^{\top } ) \over {\rm{N}} ( u_k | 0, ( \e _{k}^{(n - 1)} )^{\top } \E _2 ( \bSi ^{- 1} + \bm{T} ^{- 2} )^{- 1} {\E _2}^{\top } \e _{k}^{(n - 1)} )} d{\bta _{- (1, k + 1)}} \Big\},   
\end{align*}
where $\bta _{- (1, k + 1)}$ denotes the vector obtained by removing the first and $k+1$ th coordinates of $\bm{\eta}$ and where $( u_{k} , \bta _{- (1, k + 1)} )$ denotes $( \eta _2 , \dots , \eta _k , u_k , \eta _{k + 2} , \dots , \eta _n )^{\top } \in \mathbb{R} ^{n - 1}$. 
From this expression, we have
\begin{align*}
	&{\pd F \over \pd \z } ( \D \v ; \z )   \\
	&= \Big( \sum_{k = 1}^{n - 1} \Big\{ ( \e _{k}^{(n - 1)} )^{\top } \E _2 ( \bSi ^{- 1} + \bm{T} ^{- 2} )^{- 1} \bSi ^{- 1} \e _{i}^{(n)} \notag \\
	&\quad \times {\rm{N}} ( u_k | 0, ( \e _{k}^{(n - 1)} )^{\top } \E _2 ( \bSi ^{- 1} + \bm{T} ^{- 2} )^{- 1} {\E _2}^{\top } \e _{k}^{(n - 1)} )   \\
	&\quad \times \int_{\u _{- k} - (0, \infty )^{n - 2}} {{\rm{N}}_{n - 1} (( u_{k} , \bta _{- (1, k + 1)} ) | \bm{0} ^{(n - 1)} , \E _2 ( \bSi ^{- 1} + \bm{T} ^{- 2} )^{- 1} {\E _2}^{\top } ) \over {\rm{N}} ( u_k | 0, ( \e _{k}^{(n - 1)} )^{\top } \E _2 ( \bSi ^{- 1} + \bm{T} ^{- 2} )^{- 1} {\E _2}^{\top } \e _{k}^{(n - 1)} )} d{\bta _{- (1, k + 1)}} \notag \\
	&\quad \Big\} \Big) _{i = 1}^{n} \Big| _{\btau = \D \v }   \\
	&= \sum_{k = 1}^{n - 1} \Big\{ \bSi ^{- 1} ( \bSi ^{- 1} + \bm{T} ^{- 2} )^{- 1} {\E _2}^{\top } \e _{k}^{(n - 1)} \notag \\
	&\quad \times {\rm{N}} ( u_k | 0, ( \e _{k}^{(n - 1)} )^{\top } \E _2 ( \bSi ^{- 1} + \bm{T} ^{- 2} )^{- 1} {\E _2}^{\top } \e _{k}^{(n - 1)} )   \\
	&\quad \times \int_{\u _{- k} - (0, \infty )^{n - 2}} {{\rm{N}}_{n - 1} (( u_{k} , \bta _{- (1, k + 1)} ) | \bm{0} ^{(n - 1)} , \E _2 ( \bSi ^{- 1} + \bm{T} ^{- 2} )^{- 1} {\E _2}^{\top } ) \over {\rm{N}} ( u_k | 0, ( \e _{k}^{(n - 1)} )^{\top } \E _2 ( \bSi ^{- 1} + \bm{T} ^{- 2} )^{- 1} {\E _2}^{\top } \e _{k}^{(n - 1)} )} d{\bta _{- (1, k + 1)}} \notag \\
	&\quad \Big\} \Big| _{\btau = \D \v }   \\
	&= \sum_{k = 1}^{n - 1} \Big[ \bSi ^{- 1} ( \bSi ^{- 1} + \D ^{- 2} \V ^{- 2} )^{- 1} {\E _2}^{\top } \e _{k}^{(n - 1)} \notag \\
	&\quad \times {\rm{N}} (( u_k | _{\btau = \D \v } ) | 0, ( \e _{k}^{(n - 1)} )^{\top } \E _2 ( \bSi ^{- 1} + \D ^{- 2} \V ^{- 2} )^{- 1} {\E _2}^{\top } \e _{k}^{(n - 1)} )   \\
	&\quad \times \int_{( \u _{- k} | _{\btau = \D \v } ) - (0, \infty )^{n - 2}} \Big\{ {{\rm{N}}_{n - 1} (( u_{k} | _{\btau = \D \v } , \bta _{- (1, k + 1)} ) | \bm{0} ^{(n - 1)} , \E _2 ( \bSi ^{- 1} + \D ^{- 2} \V ^{- 2} )^{- 1} {\E _2}^{\top } ) \over {\rm{N}} (( u_k | _{\btau = \D \v } ) | 0, ( \e _{k}^{(n - 1)} )^{\top } \E _2 ( \bSi ^{- 1} + \D ^{- 2} \V ^{- 2} )^{- 1} {\E _2}^{\top } \e _{k}^{(n - 1)} )} \notag \\
	&\quad \Big\} d{\bta _{- (1, k + 1)}} \Big] \\
	&\equiv \sum_{k = 1}^{n - 1} \H _k.   
\end{align*}
Here, for all $k=1, ..., n-1$, we have
\begin{align*}
	&\| \H _k \| \notag \\
	&\sim \| \bSi ^{- 1} [ \bSi ^{- 1} + \{ \A ( \v _{- i_{\ast }} ) \} ^2 ]^{- 1} {\E _2}^{\top } \e _{k}^{(n - 1)} \|   \\
	&\quad \times {\rm{N}} (( \e _{k + 1}^{(n)} )^{\top } [ \bSi ^{- 1} + \{ \A ( \v _{- i_{\ast }} ) \} ^2 ]^{- 1} \bSi ^{- 1} \z | 0, ( \e _{k}^{(n - 1)} )^{\top } \E _2 [ \bSi ^{- 1} + \{ \A ( \v _{- i_{\ast }} ) \} ^2 ]^{- 1} {\E _2}^{\top } \e _{k}^{(n - 1)} )   \\
	&\quad \times \int_{( \u _{- k} | _{\btau = \D \v } ) - (0, \infty )^{n - 2}} \Big\{ {{\rm{N}}_{n - 1} (( u_{k} | _{\btau = \D \v } , \bta _{- (1, k + 1)} ) | \bm{0} ^{(n - 1)} , \E _2 ( \bSi ^{- 1} + \D ^{- 2} \V ^{- 2} )^{- 1} {\E _2}^{\top } ) \over {\rm{N}} (( u_k | _{\btau = \D \v } ) | 0, ( \e _{k}^{(n - 1)} )^{\top } \E _2 ( \bSi ^{- 1} + \D ^{- 2} \V ^{- 2} )^{- 1} {\E _2}^{\top } \e _{k}^{(n - 1)} )} \notag \\
	&\quad \Big\} d{\bta _{- (1, k + 1)}}   \\
	&\le \| \bSi ^{- 1} [ \bSi ^{- 1} + \{ \A ( \v _{- i_{\ast }} ) \} ^2 ]^{- 1} {\E _2}^{\top } \e _{k}^{(n - 1)} \|   \\
	&\quad \times {\rm{N}} (( \e _{k + 1}^{(n)} )^{\top } [ \bSi ^{- 1} + \{ \A ( \v _{- i_{\ast }} ) \} ^2 ]^{- 1} \bSi ^{- 1} \z | 0, ( \e _{k}^{(n - 1)} )^{\top } \E _2 [ \bSi ^{- 1} + \{ \A ( \v _{- i_{\ast }} ) \} ^2 ]^{- 1} {\E _2}^{\top } \e _{k}^{(n - 1)} )   
\end{align*}
and
$|( \e _{k + 1}^{(n)} )^{\top } [ \bSi ^{- 1} + \{ \A ( \v _{- i_{\ast }} ) \} ^2 ]^{- 1} \bSi ^{- 1} \z | \to \infty$ as $z_{i_\ast} \to \infty$. 
Therefore, we have
\begin{align*}
&{\pd F \over \pd \z } ( \D \v ; \z ) \to \bm{0}.   
\end{align*}

\noindent
\textbf{Step 5: Existence of dominance integrable functions.}
Until the previous step, we have shown $E[\bm{\eta} | \bm{z}] - \z$ is equal to the right-hand side of (\ref{basic_rep}) and calculated the limits of quantities which appear in (\ref{basic_rep}). 
In this step, we show the existence of dominance integrable functions in order to apply the dominated convergence theorem in (\ref{basic_rep}).
First, we rewrite $\{ 1 / m( \z ) \} \{ \pd / ( \pd \z ) \} m( \z )$ as
\begin{align*}
	&{1 \over m( \z )} {\pd m( \z ) \over \pd \z }   \\
	&= \int_{(0, \infty )^n} \Big[ {\prod_{i = 1}^{n} \pi (| z_i | v_i ) \over | \I _n + \D \V \bSi ^{- 1} \V \D |^{1 / 2}} \notag \\
	&\quad \times \exp \Big\{ - {1 \over 2} \w ^{\top } \V ^{- 1} ( \I _n + \D ^{- 1} \V ^{- 1} \bSi \V ^{- 1} \D ^{- 1} )^{- 1} \V^{- 1} \w \Big\} \\
	&\quad \times \Big\{ - \D ^{- 1} \V ^{- 1} ( \I _n + \D ^{- 1} \V ^{- 1} \bSi \V ^{- 1} \D ^{- 1} )^{- 1} \V ^{- 1} \w F( \D \v ; \z ) + {\pd F \over \pd \z } ( \D \v ; \z ) \Big\} \Big] d\v \\
	&\quad / \int_{(0, \infty )^n} \Big[ {\prod_{i = 1}^{n} \pi (| z_i | v_i ) \over | \I _n + \D \V \bSi ^{- 1} \V \D |^{1 / 2}} \notag \\
	&\quad \times \exp \Big\{ - {1 \over 2} \w ^{\top } \V ^{- 1} ( \I _n + \D ^{- 1} \V ^{- 1} \bSi \V ^{- 1} \D ^{- 1} )^{- 1} \V^{- 1} \w \Big\} \notag \\
	&\quad \times F( \D \v ; \z ) \Big] d\v \\
	&= \int_{(0, \infty )^n} \Big[ {\prod_{i = 1}^{n} \pi (| z_i | v_i ) \over | \bSi + \V \D ^2 \V |^{1 / 2}} \notag \\
	&\quad \times \exp \Big\{ - {1 \over 2} \w ^{\top } \V ^{- 1} ( \I _n + \D ^{- 1} \V ^{- 1} \bSi \V ^{- 1} \D ^{- 1} )^{- 1} \V^{- 1} \w \Big\} \\
	&\quad \times \Big\{ - \D ^{- 1} \V ^{- 1} ( \I _n + \D ^{- 1} \V ^{- 1} \bSi \V ^{- 1} \D ^{- 1} )^{- 1} \V ^{- 1} \w F( \D \v ; \w ) + {\pd F \over \pd \z } ( \D \v ; \z ) \Big\} \Big] d\v \\
	&\quad / \int_{(0, \infty )^n} \Big[ {\prod_{i = 1}^{n} \pi (| z_i | v_i ) \over | \bSi + \V \D ^2 \V |^{1 / 2}} \notag \\
	&\quad \times \exp \Big\{ - {1 \over 2} \w ^{\top } \V ^{- 1} ( \I _n + \D ^{- 1} \V ^{- 1} \bSi \V ^{- 1} \D ^{- 1} )^{- 1} \V^{- 1} \w \Big\} \notag \\
	&\quad \times F( \D \v ; \z ) \Big] d\v \\
	&= \int_{(0, \infty )^n} \Big[ {\prod_{i = 1}^{n} \{ \pi (| z_i | v_i ) / \pi (| z_i |) \} \over | \D ^{- 1} \bSi \D ^{- 1} + \V ^2 |^{1 / 2}} \notag \\
	&\quad \times \exp \Big\{ - {1 \over 2} \w ^{\top } \V ^{- 1} ( \I _n + \D ^{- 1} \V ^{- 1} \bSi \V ^{- 1} \D ^{- 1} )^{- 1} \V^{- 1} \w \Big\} \\
	&\quad \times \Big\{ - \D ^{- 1} \V ^{- 1} ( \I _n + \D ^{- 1} \V ^{- 1} \bSi \V ^{- 1} \D ^{- 1} )^{- 1} \V ^{- 1} \w F( \D \v ; \z ) + {\pd F \over \pd \z } ( \D \v ; \z ) \Big\} \Big] d\v \\
	&\quad / \int_{(0, \infty )^n} \Big[ {\prod_{i = 1}^{n} [ \pi (| z_i | v_i ) / \pi (| z_i |)] \over | \D ^{- 1} \bSi \D ^{- 1} + \V ^2 |^{1 / 2}} \notag \\
	&\quad \times \exp \Big\{ - {1 \over 2} \w ^{\top } \V ^{- 1} ( \I _n + \D ^{- 1} \V ^{- 1} \bSi \V ^{- 1} \D ^{- 1} )^{- 1} \V^{- 1} \w \Big\} \notag \\
	&\quad \times F( \D \v ; \z ) \Big] d\v \\
	&\equiv \int_{(0, \infty )^n} \Big[ h( \v ; z_{i_\ast} ) \notag \\
	&\quad \times \Big\{ - \D ^{- 1} \V ^{- 1} ( \I _n + \D ^{- 1} \V ^{- 1} \bSi \V ^{- 1} \D ^{- 1} )^{- 1} \V ^{- 1} \w F( \D \v ; \z ) + {\pd F \over \pd \z } ( \D \v ; \z ) \Big\} \Big] d\v \\
	&\quad / \int_{(0, \infty )^n} h( \v ; z_{i_\ast} ) F( \D \v ; \z ) d\v , 
\end{align*}
where we define
\[
h( \v ; z_{i_\ast} )
=
{\prod_{i = 1}^{n} \{ \pi (| z_i | v_i ) / \pi (| z_i |) \} \over | \D ^{- 1} \bSi \D ^{- 1} + \V ^2 |^{1 / 2}} \exp \Big\{ - {1 \over 2} \w ^{\top } \V ^{- 1} ( \I _n + \D ^{- 1} \V ^{- 1} \bSi \V ^{- 1} \D ^{- 1} )^{- 1} \V^{- 1} \w \Big\} .
\]
We will show that the functions 
$
h( \v ; z_{i_\ast} ) F( \D \v ; \z )$ and 
\[
h( \v ; z_{i_\ast} ) \Big\{ - \D ^{- 1} \V ^{- 1} ( \I _n + \D ^{- 1} \V ^{- 1} \bSi \V ^{- 1} \D ^{- 1} )^{- 1} \V ^{- 1} \w F( \D \v ; \z ) + {\pd F \over \pd \z } ( \D \v ; \z ) \Big\} 
\] 
are dominated by integrable functions on $(0, \infty)^n$ which do not depend on $z_{i_\ast}$.
First, we show this for the function $ h( \v ; z_{i_\ast} ) F( \D \v ; \z )$. Let $\varepsilon , M > 0$ such that
$\ep \I _n \le \bSi \le M \I _n$.
Observe that
\begin{align*}
{1 \over | \D ^{- 1} \bSi \D ^{- 1} + \V ^2 |^{1 / 2}} &\le {1 \over | \ep \D ^{- 2} + \V ^2 |^{1 / 2}} = \prod_{i = 1}^{n} {1 \over ( \ep / {z_i}^2 + {v_i}^2 )^{1 / 2}} = \prod_{i = 1}^{n} \Big( {{z_i}^2 \over \ep + {z_i}^2 {v_i}^2} \Big) ^{1 / 2},   
\end{align*}
\begin{align*}
\prod_{i = 1}^{n} {\pi (| z_i | v_i ) \over \pi (| z_i |)} &= \prod_{i = 1}^{n} {1 + {z_i}^2 \over 1 + {z_i}^2 {v_i}^2} \le \Big( \prod_{i \neq i_{\ast }} {1 + {z_i}^2 \over 1 + {z_i}^2 {v_i}^2} \Big) \times 2 {{z_{i_{\ast }}}^2 \over 1 + {z_{i_{\ast }}}^2 {v_{i_{\ast }}}^2} , 
\end{align*}
and
\begin{align*}
&\exp \Big\{ - {1 \over 2} \w ^{\top } \V ^{- 1} ( \I _n + \D ^{- 1} \V ^{- 1} \bSi \V ^{- 1} \D ^{- 1} )^{- 1} \V^{- 1} \w \Big\} \\
&\le \exp \Big\{ - {1 \over 2} \w ^{\top } \V ^{- 1} ( \I _n + M \D ^{- 2} \V ^{- 2} )^{- 1} \V^{- 1} \w \Big\} \\
&= \exp \Big\{ - {1 \over 2} \w ^{\top } ( \V ^2  + M \D ^{- 2} )^{- 1} \w \Big\} \notag \\
&= \prod_{i = 1}^{n} \exp \Big( - {1 \over 2} {1 \over {v_i}^2 + M / {z_i}^2} \Big) = \prod_{i = 1}^{n} \exp \Big( - {1 \over 2} {{z_i}^2 \over M + {z_i}^2 {v_i}^2} \Big).   
\end{align*}
From these facts, we obtain the following bound
\begin{align*}
&h( \v ; z_{i_\ast} ) \notag \\
&\le \Big\{ \prod_{i = 1}^{n} \Big( {{z_i}^2 \over \ep + {z_i}^2 {v_i}^2} \Big) ^{1 / 2} \Big\} \Big( \prod_{i \neq i_{\ast }} {1 + {z_i}^2 \over 1 + {z_i}^2 {v_i}^2} \Big) \Big( 2 {{z_{i_{\ast }}}^2 \over 1 + {z_{i_{\ast }}}^2 {v_{i_{\ast }}}^2} \Big) \prod_{i = 1}^{n} \exp \Big( - {1 \over 2} {{z_i}^2 \over M + {z_i}^2 {v_i}^2} \Big)   \\
&= \Big[ \prod_{i \neq i_{\ast }} \Big\{ \Big( {{z_i}^2 \over \ep + {z_i}^2 {v_i}^2} \Big) ^{1 / 2} {1 + {z_i}^2 \over 1 + {z_i}^2 {v_i}^2} \exp \Big( - {1 \over 2} {{z_i}^2 \over M + {z_i}^2 {v_i}^2} \Big) \Big\} \Big] \\
&\quad \times \Big( {{z_{i_{\ast }}}^2 \over \ep + {z_{i_{\ast }}}^2 {v_{i_{\ast }}}^2} \Big) ^{1 / 2} \Big( 2 {{z_{i_{\ast }}}^2 \over 1 + {z_{i_{\ast }}}^2 {v_{i_{\ast }}}^2} \Big) \exp \Big( - {1 \over 2} {{z_{i_{\ast }}}^2 \over M + {z_{i_{\ast }}}^2 {v_{i_{\ast }}}^2} \Big) .   
\end{align*}
Here, we have the following bound under $v_{i_{\ast }} > 1$
\begin{align*}
&\Big( {{z_{i_{\ast }}}^2 \over \ep + {z_{i_{\ast }}}^2 {v_{i_{\ast }}}^2} \Big) ^{1 / 2} \Big( 2 {{z_{i_{\ast }}}^2 \over 1 + {z_{i_{\ast }}}^2 {v_{i_{\ast }}}^2} \Big) \exp \Big( - {1 \over 2} {{z_{i_{\ast }}}^2 \over M + {z_{i_{\ast }}}^2 {v_{i_{\ast }}}^2} \Big) \le \Big( {1 \over {v_{i_{\ast }}}^2} \Big) ^{1 / 2} \Big( 2 {1 \over {v_{i_{\ast }}}^2} \Big),   
\end{align*}
and the following bound under $v_{i_{\ast }} \le 1$
\begin{align*}
&\Big( {{z_{i_{\ast }}}^2 \over \ep + {z_{i_{\ast }}}^2 {v_{i_{\ast }}}^2} \Big) ^{1 / 2} \Big( 2 {{z_{i_{\ast }}}^2 \over 1 + {z_{i_{\ast }}}^2 {v_{i_{\ast }}}^2} \Big) \exp \Big( - {1 \over 2} {{z_{i_{\ast }}}^2 \over M + {z_{i_{\ast }}}^2 {v_{i_{\ast }}}^2} \Big)   \\
&\le \Big( {M \over \ep } {{z_{i_{\ast }}}^2 \over M + {z_{i_{\ast }}}^2 {v_{i_{\ast }}}^2} \Big) ^{1 / 2} \Big( 2 M {{z_{i_{\ast }}}^2 \over M + {z_{i_{\ast }}}^2 {v_{i_{\ast }}}^2} \Big) \exp \Big( - {1 \over 2} {{z_{i_{\ast }}}^2 \over M + {z_{i_{\ast }}}^2 {v_{i_{\ast }}}^2} \Big)   \\
&\le \Big( {M \over \ep } \Big) ^{1 / 2} 2 M \Big\{ \sup_{u \in (0, \infty )} ( u^{3 / 2} e^{- u / 2} ) \Big\} \equiv M' < \infty.   
\end{align*}
Therefore, noting $F( \D \v ; \z ) \le 1$, the function 
$ h( \v ; z_{i_\ast} ) F( \D \v ; \z )$ is dominated as
\begin{align*}
h( \v ; z_{{i_\ast}} ) F( \D \v ; \z ) &\le \Big[ \prod_{i \neq i_{\ast }} \Big\{ \Big( {{z_i}^2 \over \ep + {z_i}^2 {v_i}^2} \Big) ^{1 / 2} {1 + {z_i}^2 \over 1 + {z_i}^2 {v_i}^2} \exp \Big( - {1 \over 2} {{z_i}^2 \over M + {z_i}^2 {v_i}^2} \Big) \Big\} \Big] \\
&\quad \times \Big\{ 1( v_{i_{\ast }} \le 1) M' + 1( v_{i_{\ast }} > 1) \Big( {1 \over {v_i}^2} \Big) ^{1 / 2} \Big( 2 {1 \over {v_{i_{\ast }}}^2} \Big) \Big\} .   
\end{align*}
The right-hand side of the above display is integrable on $(0, \infty)^n$ and does not depend on $z_{i_\ast}$.

Next, we show the existence of an integrable function that dominates the function
$$ 
h( \v ; z_{i_\ast} ) \Big\{ - \D ^{- 1} \V ^{- 1} ( \I _n + \D ^{- 1} \V ^{- 1} \bSi \V ^{- 1} \D ^{- 1} )^{- 1} \V ^{- 1} \w F( \D \v ; \z ) + {\pd F \over \pd \z } ( \D \v ; \z ) \Big\} . $$
For a sufficiently large $M_2 > 0$, we have the following bound
\begin{align*}
	&\| - \D ^{- 1} \V ^{- 1} ( \I _n + \D ^{- 1} \V ^{- 1} \bSi \V ^{- 1} \D ^{- 1} )^{- 1} \V ^{- 1} \w F( \D \v ; \z ) \|   \\
	&\le \| \D ^{- 1} \| \| \V ^{- 1} ( \I _n + \D ^{- 1} \V ^{- 1} \bSi \V ^{- 1} \D ^{- 1} )^{- 1} \V ^{- 1} \| \| \w \|   \\
	&= \sqrt{\sum_{i = 1}^{n} {1 \over {z_i}^2}} \| ( \V ^2 + \D ^{- 1} \bSi \D ^{- 1} )^{- 1} \| \sqrt{n}   \\
	&\le \sqrt{\sum_{i = 1}^{n} {n \over {z_i}^2}} \sum_{i = 1}^{n} \sum_{j = 1}^{n} | ( \e _{i}^{(n)} )^{\top } ( \V ^2 + \D ^{- 1} \bSi \D ^{- 1} )^{- 1} \e _{j}^{(n)} |   \\
	&\le \sqrt{\sum_{i = 1}^{n} {n \over {z_i}^2}} \sum_{i = 1}^{n} \sum_{j = 1}^{n} {( \e _{i}^{(n)} )^{\top } ( \V ^2 + \D ^{- 1} \bSi \D ^{- 1} )^{- 1} \e _{i}^{(n)} + ( \e _{j}^{(n)} )^{\top } ( \V ^2 + \D ^{- 1} \bSi \D ^{- 1} )^{- 1} \e _{j}^{(n)} \over 2}   \\
	&\le \sqrt{\sum_{i = 1}^{n} {n \over {z_i}^2}} \sum_{i = 1}^{n} \sum_{j = 1}^{n} {( \e _{i}^{(n)} )^{\top } ( \V ^2 + \ep \D ^{- 2} )^{- 1} \e _{i}^{(n)} + ( \e _{j}^{(n)} )^{\top } ( \V ^2 + \ep \D ^{- 2} )^{- 1} \e _{j}^{(n)} \over 2}   \\
	&= \sqrt{\sum_{i = 1}^{n} {n \over {z_i}^2}} {1 \over 2} \sum_{i = 1}^{n} \sum_{j = 1}^{n} \Big( {{z_i}^2 \over \ep + {z_i}^2 {v_i}^2} + {{z_j}^2 \over \ep + {z_j}^2 {v_j}^2} \Big) = \sqrt{\sum_{i = 1}^{n} {n \over {z_i}^2}} n \sum_{i = 1}^{n} {{z_i}^2 \over \ep + {z_i}^2 {v_i}^2} \notag \\
	&\le M_2 \sum_{i = 1}^{n} {{z_i}^2 \over \ep + {z_i}^2 {v_i}^2}.   
\end{align*}
Therefore, we obtain
\begin{align*}
	&\| - \D ^{- 1} \V ^{- 1} ( \I _n + \D ^{- 1} \V ^{- 1} \bSi \V ^{- 1} \D ^{- 1} )^{- 1} \V ^{- 1} \w F( \D \v ; \z ) h( \v ; z_{i_\ast} ) \|   \\
	&\le M_2 \Big( \sum_{i' \neq i_{\ast }} {{z_{i'}}^2 \over \ep + {z_{i'}}^2 {v_{i'}}^2} + {{z_{i_{\ast }}}^2 \over \ep + {z_{i_{\ast }}}^2 {v_{i_{\ast }}}^2} \Big) \notag \\
	&\quad \times \Big[ \prod_{i \neq i_{\ast }} \Big\{ \Big( {{z_i}^2 \over \ep + {z_i}^2 {v_i}^2} \Big) ^{1 / 2} {1 + {z_i}^2 \over 1 + {z_i}^2 {v_i}^2} \exp \Big( - {1 \over 2} {{z_i}^2 \over M + {z_i}^2 {v_i}^2} \Big) \Big\} \Big] \\
	&\quad \times \Big( {{z_{i_{\ast }}}^2 \over \ep + {z_{i_{\ast }}}^2 {v_{i_{\ast }}}^2} \Big) ^{1 / 2} \Big( 2 {{z_{i_{\ast }}}^2 \over 1 + {z_{i_{\ast }}}^2 {v_{i_{\ast }}}^2} \Big) \exp \Big( - {1 \over 2} {{z_{i_{\ast }}}^2 \over M + {z_{i_{\ast }}}^2 {v_{i_{\ast }}}^2} \Big) \\
	&= M_2 \sum_{i' \neq i_{\ast }} \Big[ \prod_{i \neq i_{\ast }} \Big\{ \Big( {{z_i}^2 \over \ep + {z_i}^2 {v_i}^2} \Big) ^{1 / 2 + 1(i = i' )} {1 + {z_i}^2 \over 1 + {z_i}^2 {v_i}^2} \exp \Big( - {1 \over 2} {{z_i}^2 \over M + {z_i}^2 {v_i}^2} \Big) \Big\} \Big] \\
	&\quad \times \Big( {{z_{i_{\ast }}}^2 \over \ep + {z_{i_{\ast }}}^2 {v_{i_{\ast }}}^2} \Big) ^{1 / 2} \Big( 2 {{z_{i_{\ast }}}^2 \over 1 + {z_{i_{\ast }}}^2 {v_{i_{\ast }}}^2} \Big) \exp \Big( - {1 \over 2} {{z_{i_{\ast }}}^2 \over M + {z_{i_{\ast }}}^2 {v_{i_{\ast }}}^2} \Big) \\
	&\quad + M_2 \Big[ \prod_{i \neq i_{\ast }} \Big\{ \Big( {{z_i}^2 \over \ep + {z_i}^2 {v_i}^2} \Big) ^{1 / 2} {1 + {z_i}^2 \over 1 + {z_i}^2 {v_i}^2} \exp \Big( - {1 \over 2} {{z_i}^2 \over M + {z_i}^2 {v_i}^2} \Big) \Big\} \Big] \\
	&\quad \times \Big( {{z_{i_{\ast }}}^2 \over \ep + {z_{i_{\ast }}}^2 {v_{i_{\ast }}}^2} \Big) ^{1 / 2 + 1} \Big( 2 {{z_{i_{\ast }}}^2 \over 1 + {z_{i_{\ast }}}^2 {v_{i_{\ast }}}^2} \Big) \exp \Big( - {1 \over 2} {{z_{i_{\ast }}}^2 \over M + {z_{i_{\ast }}}^2 {v_{i_{\ast }}}^2} \Big) \\
	&\le M_2 \sum_{i' \neq i_{\ast }} \Big[ \prod_{i \neq i_{\ast }} \Big\{ \Big( {{z_i}^2 \over \ep + {z_i}^2 {v_i}^2} \Big) ^{1 / 2 + 1(i = i' )} {1 + {z_i}^2 \over 1 + {z_i}^2 {v_i}^2} \exp \Big( - {1 \over 2} {{z_i}^2 \over M + {z_i}^2 {v_i}^2} \Big) \Big\} \Big] \\
	&\quad \times \Big\{ 1( v_{i_{\ast }} \le 1) M' + 1( v_{i_{\ast }} > 1) \Big( {1 \over {v_{i_{\ast }}}^2} \Big) ^{1 / 2} \Big( 2 {1 \over {v_{i_{\ast }}}^2} \Big) \Big\}   \\
	&\quad + M_2 \Big[ \prod_{i \neq i_{\ast }} \Big\{ \Big( {{z_i}^2 \over \ep + {z_i}^2 {v_i}^2} \Big) ^{1 / 2} {1 + {z_i}^2 \over 1 + {z_i}^2 {v_i}^2} \exp \Big( - {1 \over 2} {{z_i}^2 \over M + {z_i}^2 {v_i}^2} \Big) \Big\} \Big] \\
	&\quad \times \Big\{ 1( v_{i_{\ast }} \le 1) M'' + 1( v_{i_{\ast }} > 1) \Big( {1 \over {v_{i_{\ast }}}^2} \Big) ^{1 / 2 + 1} \Big( 2 {1 \over {v_{i_{\ast }}}^2} \Big) \Big\} 
\end{align*}
for some $M'' > 0$, where the right hand side of the above display is an integrable function of $\bm{v}$ on $(0, \infty )^n$ which does not depend on $z_{i_\ast}$.
Furthermore, $\{ ( \pd F) / ( \pd \z ) \} ( \D \v ; \z )$ is bounded as
\begin{align*}
&\Big\| {\pd F \over \pd \z } ( \D \v ; \z ) \Big\| \le \sum_{k = 1}^{n - 1} \| \H _k \|   \\
&\le \sum_{k = 1}^{n - 1} \{ \| \bSi ^{- 1} ( \bSi ^{- 1} + \D ^{- 2} \V ^{- 2} )^{- 1} {\E _2}^{\top } \e _{k}^{(n - 1)} \| \notag \\
&\quad \times {\rm{N}} (( u_k | _{\bla = \D \v } ) | 0, ( \e _{k}^{(n - 1)} )^{\top } \E _2 ( \bSi ^{- 1} + \D ^{- 2} \V ^{- 2} )^{- 1} {\E _2}^{\top } \e _{k}^{(n - 1)} ) \} \\
&\le {1 \over \sqrt{2 \pi }} \sum_{k = 1}^{n - 1} {\| \bSi ^{- 1} ( \bSi ^{- 1} + \D ^{- 2} \V ^{- 2} )^{- 1} {\E _2}^{\top } \e _{k}^{(n - 1)} \| \over \sqrt{( \e _{k}^{(n - 1)} )^{\top } \E _2 ( \bSi ^{- 1} + \D ^{- 2} \V ^{- 2} )^{- 1} {\E _2}^{\top } \e _{k}^{(n - 1)}}}   \\
&= {1 \over \sqrt{2 \pi }} \sum_{k = 1}^{n - 1} \Big\{ {1 \over \sqrt{( \e _{k}^{(n - 1)} )^{\top } \E _2 ( \bSi ^{- 1} + \D ^{- 2} \V ^{- 2} )^{- 1} {\E _2}^{\top } \e _{k}^{(n - 1)}}} \notag \\
&\quad \times \sqrt{( \e _{k}^{(n - 1)} )^{\top } \E _2 ( \bSi ^{- 1} + \D ^{- 2} \V ^{- 2} )^{- 1} \bSi ^{- 1 / 2} \bSi ^{- 1} \bSi ^{- 1 / 2} ( \bSi ^{- 1} + \D ^{- 2} \V ^{- 2} )^{- 1} {\E _2}^{\top } \e _{k}^{(n - 1)}} \Big\} \notag \\
&\le {\sqrt{1 / \ep } \over \sqrt{2 \pi }} \sum_{k = 1}^{n - 1} {\sqrt{( \e _{k}^{(n - 1)} )^{\top } \E _2 ( \bSi ^{- 1} + \D ^{- 2} \V ^{- 2} )^{- 1} \bSi ^{- 1} ( \bSi ^{- 1} + \D ^{- 2} \V ^{- 2} )^{- 1} {\E _2}^{\top } \e _{k}^{(n - 1)}} \over \sqrt{( \e _{k}^{(n - 1)} )^{\top } \E _2 ( \bSi ^{- 1} + \D ^{- 2} \V ^{- 2} )^{- 1} {\E _2}^{\top } \e _{k}^{(n - 1)}}}   \\
&\le {\sqrt{1 / \ep } \over \sqrt{2 \pi }} \sum_{k = 1}^{n - 1} {\sqrt{( \e _{k}^{(n - 1)} )^{\top } \E _2 ( \bSi ^{- 1} + \D ^{- 2} \V ^{- 2} )^{- 1} {\E _2}^{\top } \e _{k}^{(n - 1)}} \over \sqrt{( \e _{k}^{(n - 1)} )^{\top } \E _2 ( \bSi ^{- 1} + \D ^{- 2} \V ^{- 2} )^{- 1} {\E _2}^{\top } \e _{k}^{(n - 1)}}} = {\sqrt{1 / \ep } \over \sqrt{2 \pi }} (n - 1).   
\end{align*}
Therefore, we obtained the desired result.

\noindent
\textbf{Step 6: Conclusion.}
By the dominated convergence theorem, we have
\begin{align*}
	&{1 \over m( \z )} {\pd m( \z ) \over \pd \z }   \\
	&\sim \int_{(0, \infty )^n} \Big\{ {\big\{ \prod_{i \neq i_{\ast }} \pi (| z_i | v_i ) \big\} \pi (| z_{i_{\ast }} |) / {v_{i_{\ast }}}^2 \over | \bSit + \V ^2 |^{1 / 2}} \exp \Big\{ - {1 \over 2} \w ^{\top } ( \V ^2 + \bSit )^{- 1} \w \Big\} \\
	&\quad \times \Big( - \A ( \v _{- i_{\ast }} ) \{ \I _n + \A ( \v _{- i_{\ast }} ) \bSi \A ( \v _{- i_{\ast }} ) \} ^{- 1} \V ^{- 1} \w   \\
	&\quad \times P( {\rm{N}}_{n - 2} ( \bm{0} ^{(n - 2)} , \E _{- (1, i_{\ast } )} [ \bSi ^{- 1} + \{ \A ( \v _{- i_{\ast }} ) \} ^2 ]^{- 1} {\E _{- (1, i_{\ast } )}}^{\top } ) \\
	&\quad \in \z _{- (1, i_{\ast } )} - \E _{- (1, i_{\ast } )} \{ \B ( \v _{- i_{\ast }} ) + \V \bSi ^{- 1} \} ^{- 1} \V ^{- 1} \wbt - (0, \infty )^{n - 2} ) \Big) \Big\} d\v \\
	&\quad / \int_{(0, \infty )^n} \Big( {\big\{ \prod_{i \neq i_{\ast }} \pi (| z_i | v_i ) \big\} \pi (| z_{i_{\ast }} |) / {v_{i_{\ast }}}^2 \over | \bSit + \V ^2 |^{1 / 2}} \exp \Big\{ - {1 \over 2} \w ^{\top } ( \V ^2 + \bSit )^{- 1} \w \Big\} \\
	&\quad \times P( {\rm{N}}_{n - 2} ( \bm{0} ^{(n - 2)} , \E _{- (1, i_{\ast } )} [ \bSi ^{- 1} + \{ \A ( \v _{- i_{\ast }} ) \} ^2 ]^{- 1} {\E _{- (1, i_{\ast } )}}^{\top } ) \\
	&\quad \in \z _{- (1, i_{\ast } )} - \E _{- (1, i_{\ast } )} \{ \B ( \v _{- i_{\ast }} ) + \V \bSi ^{- 1} \} ^{- 1} \V ^{- 1} \wbt - (0, \infty )^{n - 2} ) \Big) d\v \\
	&\equiv - \int_{(0, \infty )^n} \A ( \v _{- i_{\ast }} ) \{ \I _n + \A ( \v _{- i_{\ast }} ) \bSi \A ( \v _{- i_{\ast }} ) \} ^{- 1} \V ^{- 1} \w h( \v ) d\v / \int_{(0, \infty )^n} h( \v ) d\v . 
\end{align*}
Therefore, we have
\begin{align*}
&E[\bm{\eta} | \bm{z}] - \z \notag \\
&\sim - \bSi \int_{(0, \infty )^n} \A ( \v _{- i_{\ast }} ) \{ \I _n + \A ( \v _{- i_{\ast }} ) \bSi \A ( \v _{- i_{\ast }} ) \} ^{- 1} \V ^{- 1} \w h( \v ) d\v / \int_{(0, \infty )^n} h( \v ) d\v , 
\end{align*}
as $z_{i_\ast} \to \infty$. 
This implies $E[\bm{\eta} | \bm{z}] - \z$ is bounded, so the claim of the theorem follows.

\section*{A3 \,\ Proof of Theorem 
2}
In the proof of Theorem 
2, we use the following lemma, which gives a relationship between the Ces\'{a}ro-average risk and the prior $p$. For any measurable set $A \subset [0, \infty)$, we denote $p(A)$ as the prior measure.

\begin{lem}[\cite{clarke1990information}]
	For any $\varepsilon > 0$, let $A_\varepsilon = \{\eta \ge 0: \Delta_{KL}(\eta_0 || \eta) < \varepsilon \}$ denote the Kullback-Leibler information neighborhood of size $\varepsilon$, centered at  $\eta_0$. Then the following upper bound for $R_n$ holds for all $\varepsilon > 0$ and $n$: 
	\begin{equation}
	R_n \le 
	\varepsilon - \frac{1}{n} \log p(A_\varepsilon).
	\label{lem:bound}
	\end{equation}
	\label{risk_kl}
\end{lem}

\begin{proof}[
Proof of Theorem 2]
	First, we consider the case $\eta_0 = 0$. In this case, the Kullback-Leibler information neighborhood is 
	$A_\varepsilon = [0, \sqrt{2\varepsilon})$. Using the lower bounds for the horseshoe density (Theorem 1 in \cite{carvalho2010horseshoe}), we have the following lower bound for $p(A_\varepsilon)$:
	\begin{align*}
	p(A_\varepsilon)
	&\ge 
	\frac{1}{(2\pi^3)^{1/2}} \int_0^{\sqrt{2\varepsilon}} \log\left(1 + \frac{4}{\eta^2}\right) d\eta  \\
	&=
	\frac{1}{(2\pi^3)^{1/2}} \int_{2/\varepsilon}^{\infty} \frac{\log(1 + u)}{u^{3/2}} du  \\
	&\ge
	C_1\varepsilon^{1/2}\log \left(1 + \frac{1}{\varepsilon}\right),
	\end{align*}
	where $C_1$ is a constant. 
	Setting $\varepsilon = 1/n$ and applying Lemma \ref{risk_kl} gives the bound
	\[
	R_n 
	=
	O\left(\frac{1}{n} + \frac{1}{2n} \log n - \frac{1}{n} \log \log (1 + n)\right)
	=
	O[n^{-1}\{\log n - \log \log n \}].
	\]
	
	Next, we consider the case $\eta_0 > 0$.
	In this case, for all small $\varepsilon > 0$, the Kullback-Leibler information neighborhood is $A_\varepsilon = (\eta_0 - \sqrt{2\varepsilon}, \eta_0 + \sqrt{2\varepsilon})$. 
	Because the half-horseshoe density is bounded below by a constant $C_2$ near $\eta_0$, we have 
	\[
	p(A_\varepsilon)
	\ge 
	\int_{\eta_0 - \sqrt{2\varepsilon}}^{\eta_0 + \sqrt{2\varepsilon}} C_2d\eta
	=
	2\sqrt{2}C_2\sqrt{\varepsilon}.
	\]
	Setting
	$\varepsilon = 1/n$ and applying Lemma \ref{risk_kl} gives the bound
	\[
	R_n = 
	O\left(\frac{1}{n} + \frac{1}{2n}\log n \right)
	=
	O\{n^{-1}\log n\}.
	\]
\end{proof}

\section*{A4 \,\ Additional simulation results}
We here provide additional simulation results.
First, we reported example fits using the HL, HN and GP methods in Figure \ref{fig:example_fits_2}. We see that the HL and HN methods are less adaptive to abrupt increases than the HH method in Figure 2. We also see that the estimated values by the GP method are not always monotonic.

Next, focusing the comparison on the three methods based on shrinkage priors ( HH, HL and HN),  we reported boxplots of simulation results by the three methods in Figure \ref{fig:box_plot}. It is observed that in scenarios in (I) and (II), where the true function $f$ has a less frequent increase, the HH method tends to have the shortest credible intervals while the HN method has the longest. This result would show that the strong shrinkage effects of the half-horseshoe prior were realized for many zero signals. On the other hand, in scenarios (III) and (IV) where $f$ has more frequent increments, the HH method tends to have the longest credible intervals while the HN has the shortest. This would show that the heavy tail property of the half-horseshoe prior realized for many positive signals.

\begin{figure}[htbp]
	\centering
	\includegraphics[keepaspectratio, width=5.00in]{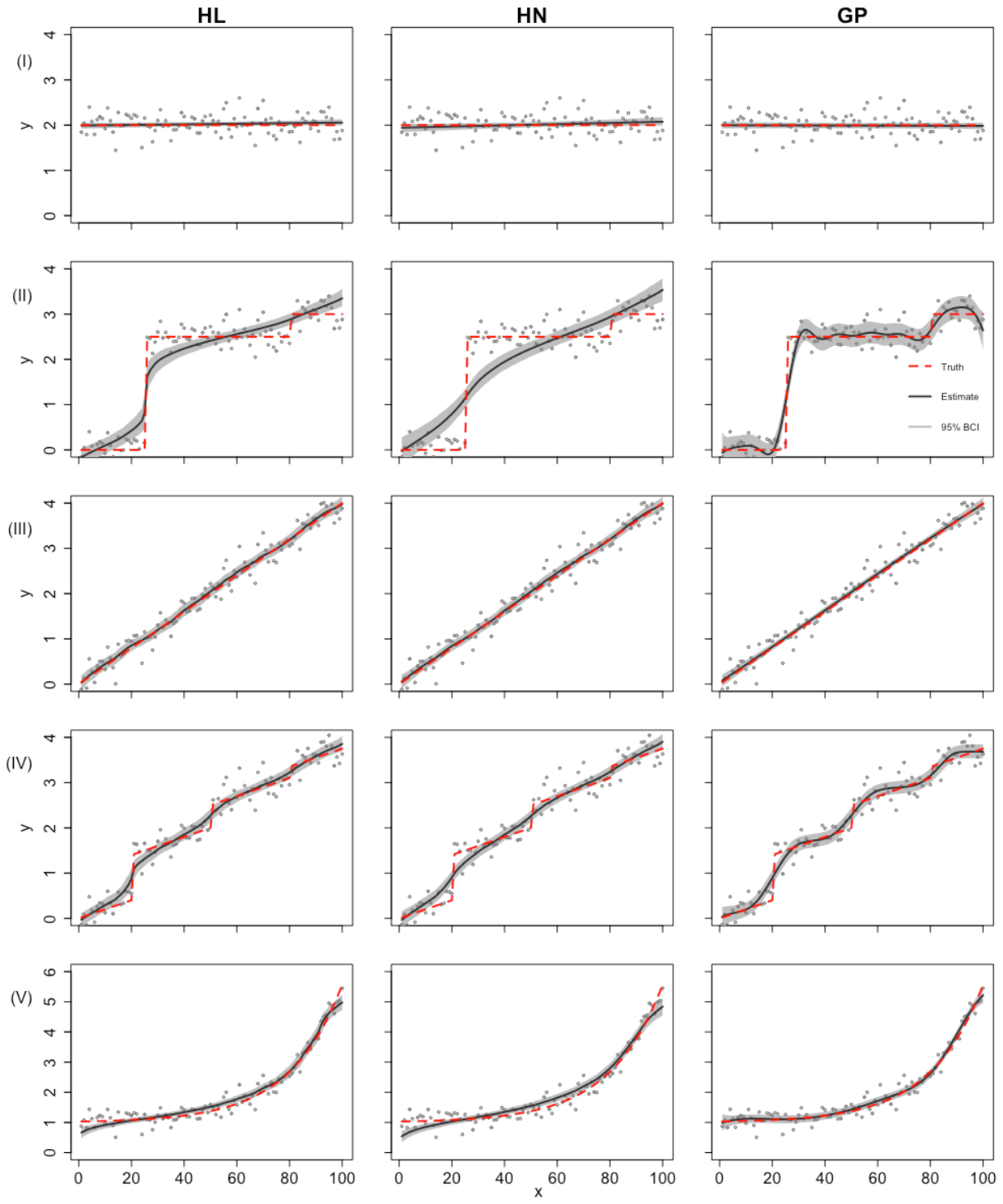}
	\caption{Example fits by the HL, HN and GP methods under the five scenarios. Plots show true functions (dashed red lines), posterior means(solid black lines), and associated $95 \%$ Bayesian credible interval (gray bands) for each $\theta_i$. Values between observed locations are interpolated for plotting.}
	\label{fig:example_fits_2}
\end{figure}

\begin{figure}[htbp]
	\centering
	\includegraphics[keepaspectratio, width=5.65in]{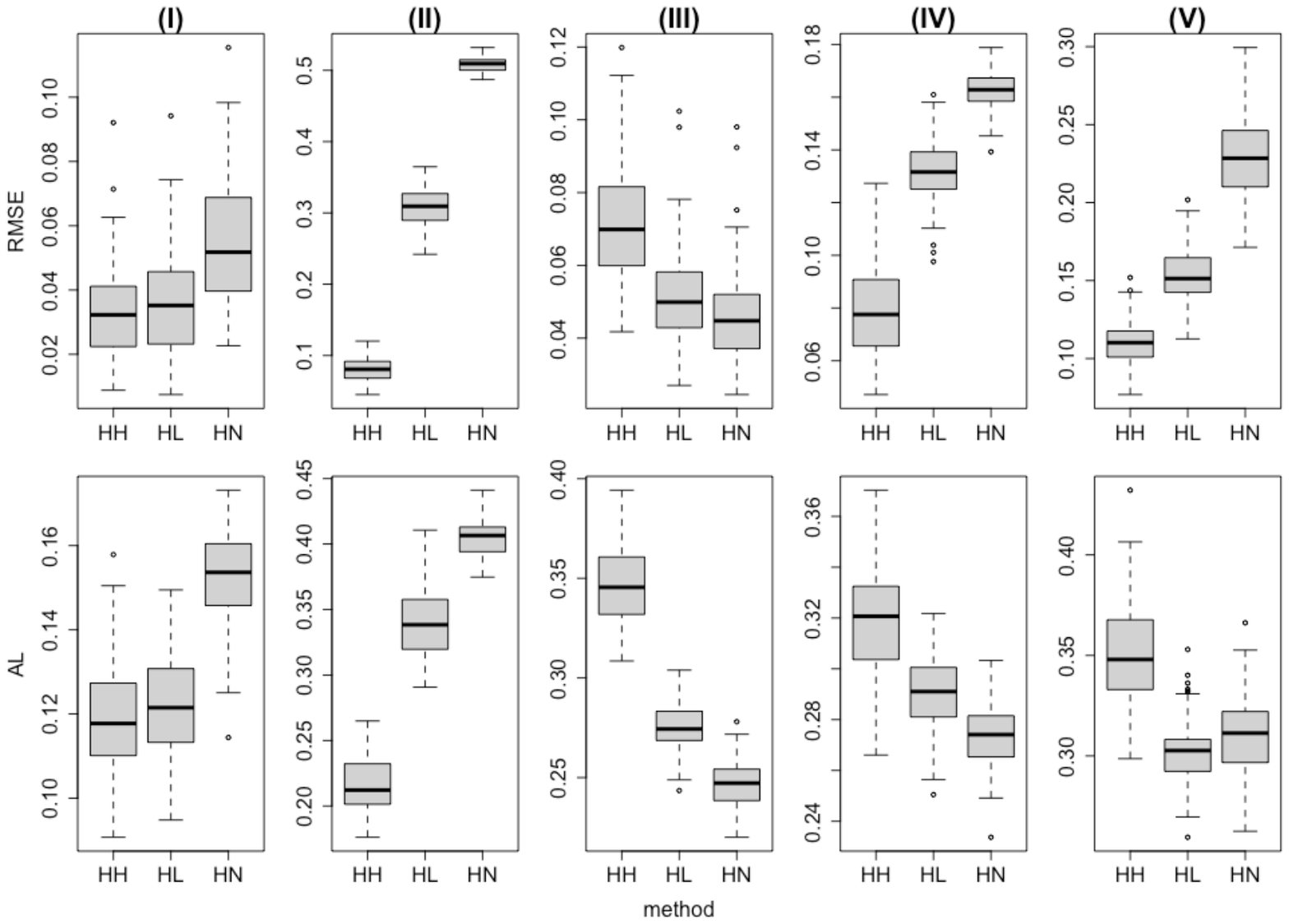}
	\caption{Boxplots of simulation results by the HH, HL and HN methods under five scenarios.  The first row shows root mean square errors (RMSEs) of point estimators and the second row shows average lengths (ALs) of $95\%$ credible intervals. }
	\label{fig:box_plot}
\end{figure}

Finally, we reported coverage probabilities (CPs) of credible intervals at different locations in Table \ref{tab:covarage}. In the five scenarios, we estimated the function values by the proposed HH method 100 times and calculated CPs of the $95\%$ credible intervals at $i \in \{5, 10, 15, 20, 25, 30, 35, 40, 45, 50\}$.

\begin{table}[!htbp]
	\begin{center}
		\begin{tabular}{ccccccccccc} \toprule
			$i$ & 5 & 10 & 15 & 20 & 25 & 30 & 35 & 40 & 45 & 50   \\ \midrule
			(I) & 0.82 &  0.84& 0.85& 0.88& 0.93& 0.95& 0.95& 0.97& 0.99& 0.98     \\
			(II) & 0.98 & 0.96& 0.93& 0.87& 0.74& 0.53& 0.66& 0.86& 0.90& 0.98  \\
			(III) & 0.98& 0.94& 0.97& 0.98& 0.96& 0.99& 0.99& 0.97& 0.97& 0.98  \\
			(IV) & 0.98& 0.97& 0.96& 0.91& 0.95& 0.96& 0.99& 0.99& 0.95& 0.89 \\
			(V) & 0.47& 0.77& 0.93& 0.96& 0.95& 0.94& 0.91& 0.90& 0.88& 0.93  \\
			\bottomrule
		\end{tabular}
		\caption{Coverage probabilities of $95\%$ credible intervals by the HH method at different locations}
		\label{tab:covarage}
	\end{center}
\end{table}

\section*{A5 \,\ Convergence diagnosis of Gibbs sampler}

Figures~\ref{fig:sample_path_theta} and \ref{fig:sample_path_others} show the sample paths of the selected model parameters used in the posterior analysis under the half-horseshoe model in Scenario (II) of Section~4. Overall, the sampler is successful in exploring the parameter space efficiently, except that one could view the sampling of $\lambda$ inefficiently. 
The mean of the effective sample sizes of function values $\theta_i$'s is 432.2786, and its standard deviation is 
302.2642. The mean of the effective sample sizes of all parameters is 939.3129, and its standard deviation is 
660.5398. We have also increased the number of MCMC iterations from 2500 to 20000 to see a possible difference in posterior analysis. As summarized in Table~\ref{table:num_pos_samp}, there is little or no difference in the RMSEs, CPs and ALs. 

\begin{figure}
    \centering
\includegraphics[keepaspectratio, width=5.65in]{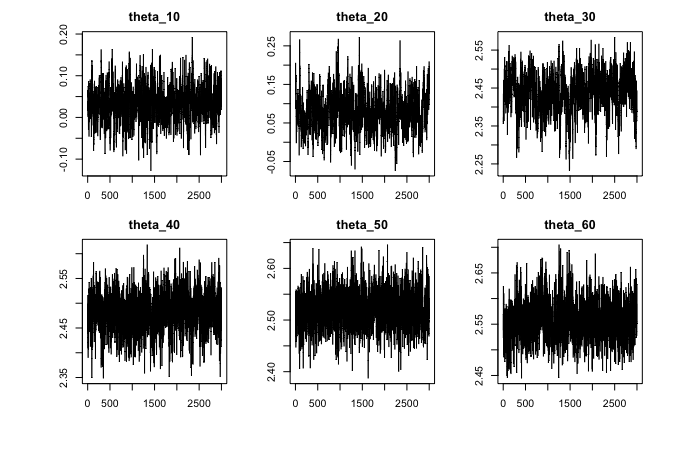}
    \caption{Paths of 2500 MCMC samples (with 500 burn-in) of $\theta _{10}$, $\theta _{20}$, $\theta _{30}$, $\theta _{40}$, $\theta _{50}$ and $\theta _{60}$ under the half-horseshoe priors in Scenario (II) of Section~4.}
    \label{fig:sample_path_theta}
\end{figure}

\begin{figure}
    \centering
\includegraphics[keepaspectratio, width=5.65in]{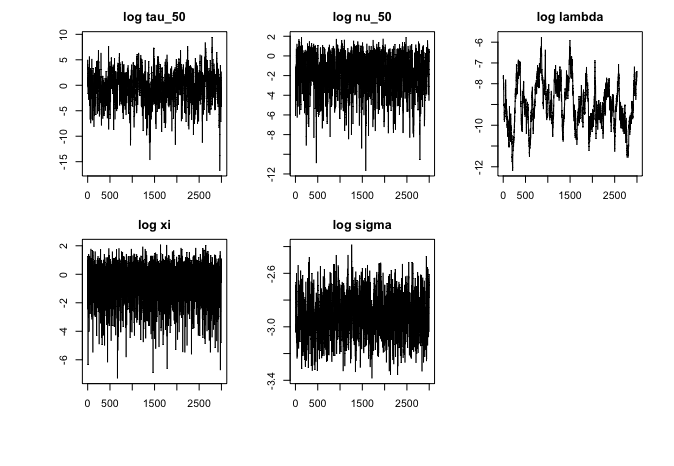}
    \caption{Sample paths of $\tau_{50}^2$, $\nu_{50}^2$, $\lambda_{50}^2$, $\xi^2$ and $\sigma^2$ in the log-scale (half-horseshoe, Scenario (II)).}
    \label{fig:sample_path_others}
\end{figure}

\begin{table}[!htbp]
	\centering
	\begin{tabular}{ccc} \toprule
		 MCMC iterations & 2500   & 20000  \\ \midrule
		 RMSE  &0.089   &0.089   \\
		CP    &90.8   &89.8    \\
		AL    &0.210   & 0.211    \\ \bottomrule
	\end{tabular}
	\caption{The averaged values of the RMSEs, CPs and ALs with different numbers of posterior samples.}
	\label{table:num_pos_samp}
\end{table}

\section*{A6 \,\ Predictive analysis}

Using the Nile River data in Section~5, we illustrate the predictive analysis by the isotonic regression models with the half shrinkage priors. Unlike the plug-in approach in Section~4, we take the formal, fully-Bayesian approach by computing the one-step ahead predictive distribution. In doing so, we repeatedly employ the MCMC method using $y_{1:n}$ at every $n\ge 20$ to forecast $y_{n+1}$. The predictive distribution, $p(y_{n+1}|y_{1:n})$, can be computed easily by adding the following step to the MCMC algorithm in Section~A1. 
\begin{itemize}
    \item[-] {\bf (Sampling of $y_{n+1}$)} \ \  Conditional on the sampled values of $\sigma^2$, $\lambda^2$ and $\theta_n$, 
    \begin{enumerate}
        \item Generate $\tau _{n+1}$ from the Cauchy distribution (the half-horseshoe model) or generate $\tau_{n+1}^2$ from the exponential distribution (the half-Laplace model), or set $\tau_{n+1}=1$ (the half-normal model). 
        \item Generate $\eta _{n+1}$ from $N_+(0,\sigma ^2\lambda ^2 \tau_{n+1}^2)$.
        \item Set $\theta _{n+1} = \eta _{n+1} + \theta_n$.
        \item Generate $y_{t+1}$ from $N(\theta_{n+1},\sigma^2)$.
    \end{enumerate}
\end{itemize}
The predictive analysis can be done by using the generated samples of $y_{n+1}$. We computed and plotted the predictive posterior means and 95\% credible intervals using the half-horseshoe, half-Laplace and half-normal priors in Figure~\ref{nile_pred}. It is seen in this figure that the point forecasts of the half-horseshoe model exhibit several temporal jumps. In addition, the predictive intervals of the half-horseshoe model are narrower than those of the other models. Table~\ref{table:forecast} summarizes the predictive results by the RMSEs of the posterior means and the empirical coverage rates and average lengths of the 95\% predictive intervals. While the half-horseshoe model has the best RMSE, its coverage rate is below the nominal level and could underestimate the predictive uncertainty. 

\begin{figure}[htbp]
	\centering
	\includegraphics[keepaspectratio, width=5.65in]{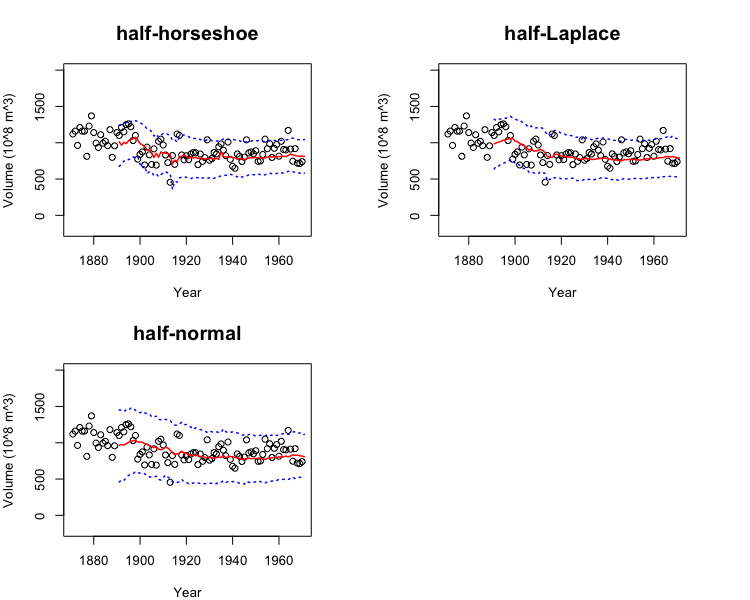}
	\caption{Predictive posterior means (solid, red) and $95\%$ credible intervals (dashed, blue) of $y_{n+1}$. }
	\label{nile_pred}
\end{figure}

\begin{table}[!htbp]
	\centering
	\begin{tabular}{ccccc} \toprule
		        &   & HH   & HL & HN  \\ \midrule
		 & RMSE  &146.212  & 154.253&  152.502 \\
		& CP    &91.2   &93.7& 97.5  \\
		 & AL    &510.380   &575.328  &735.874   \\ \bottomrule
	\end{tabular}
	\caption{The averaged values of the RMSEs, CPs and ALs of the one-step ahead predictive analysis.}
	\label{table:forecast}
\end{table}

\section*{A7 \,\ Choice of half-shrinkage priors}

As seen in the numerical examples of our research, there is a significant difference between the half-normal/Laplace/horseshoe priors in their posterior and predictive analyses. Choosing an appropriate half-shrinkage prior, depending on the context, is an important question for practitioners. To complement such an important decision, one can compare the three half-shrinkage priors by evaluating model-selection measures that can be computed by using the MCMC samples. To exemplify this approach to the choice of priors, we compute the posterior predictive losses (PPL; \citealt{gelfand1998model}) in each scenario of Section~4 and the Nile River data analysis in Section~5, which are summarized in Table~\ref{table:ppl}. For the simulation data of Scenario (II) and Nile river data, where jumps are observed and the half-horseshoe model is expected to fit well, the PPL of the half-horseshoe model is the smallest. By contrast, in the example of piecewise linear functions of Scenario (III), the half-normal and half-Laplace models show smaller PPLs. These observations support the conclusions we made in Section~4 and 5.

\begin{table}[!htbp]
	\centering
	\begin{tabular}{ccccccc} \toprule
            & (I) & (II) & (III) & (IV) & (V) & Nile ($\times 10^6$) \\ \midrule
	    HH  &15.70 & 18.39 &12.39 &14.59 &17.19 & 9.63 \\
        HL  &29.40 & 45.90 &11.94 &16.58 &28.62 & 13.75 \\
	    HN  &36.96 & 67.57 &11.88 &17.82 &32.78 & 18.01 \\ \bottomrule
	\end{tabular}
	\caption{The posterior predictive losses (PPL) of the half-horseshoe (HH), half-Laplace (HL) and half-normal (HN) models computed by using the five simulation data in Section~4 and the Nile River data in Section~5.}
	\label{table:ppl}
\end{table}

\vspace{1cm}
\bibliographystyle{chicago}
\bibliography{ref}

\end{document}